\newif\ifdraft\drafttrue
\documentclass[conference]{IEEEtran}

\usepackage{amsmath,amsthm,amsfonts,amssymb,stmaryrd,mathrsfs}
\usepackage{bbm}
\usepackage{bussproofs}
\usepackage{graphics}
\usepackage[all]{xy}
\usepackage{cite}

\usepackage{listings}
\usepackage{cleveref}
\crefname{section}{\S}{\S\S}
\crefname{subsection}{\S}{\S\S}
\crefname{subsubsection}{\S}{\S\S}
\usepackage{xcolor}
\usepackage{mdframed}
\usepackage{multirow}

\newenvironment{comment}%
  {\begin{mdframed}[linecolor=white,backgroundcolor=lightgray]}%
  {\end{mdframed}}

\usepackage{soul}
\usepackage[normalem]{ulem}
\usepackage{enumitem}

\theoremstyle{plain}
\newtheorem{theorem}{Theorem}

\newtheorem{proposition}[theorem]{Proposition}

\newtheorem{example}[theorem]{Example}
\theoremstyle{definition}

\newtheorem{remark}[theorem]{Remark}

\newcommand{\pFont}[1]{\mathtt{#1}}
\renewcommand{\tt}[1]{\pFont{#1}}
\newcommand{\Roban}{\mathbf{RoBan}}

\newcommand{\cMeas}{\mathbf{Meas}}
\newcommand{\Krn}{\mathbf{Krn}}

\newcommand{\mMeas}{\mathscr{M}}
\newcommand{\Meas}{\mathcal{M}}

\newcommand{\N}{\mathbb{N}}
\newcommand{\R}{\mathbb{R}}

\newcommand{\PosDef}{\mathrm{PosDef}}
\newcommand{\Lps}[1][p]{\mathrm{L}_{#1}}
\newcommand{\one}{\mathbbm{1}}
\newcommand{\salg}[1][F]{\mathcal{#1}}
\newcommand{\Type}[1][T]{\pFont{#1}}
\newcommand{\tBool}{\pFont{bool}}
\newcommand{\tPosDef}[1][n]{\pFont{PosDef}(#1)}
\newcommand{\tReal}{\pFont{real}}
\newcommand{\tInt}{\pFont{int}}
\newcommand{\tMeas}{\pFont{M}}
\newcommand{\tUnit}{\pFont{unit}}
\newcommand{\tGamma}{\pFont{\Gamma}}
\newcommand{\context}[1][\Gamma]{\pFont{#1}\vdash}
\newcommand{\op}{^{\mathrm{op}}}
\newcommand{\inv}{^{-1}}
\newcommand{\supp}{\mathrm{supp}}
\newcommand{\kd}{^{\sigma}}
\newcommand\BB{{\cal B}}
\newcommand\set[2]{\{#1\mid#2\}}
\newcommand{\absv}[1]{\vert #1\vert}
\newcommand{\norm}[1]{\Vert #1\Vert}
\newcommand{\cint}{\hspace{-3pt}\int\hspace{-3pt}}  
\newcommand{\tp}{\otimes} 
\newcommand{\ptp}{~\widehat{\scalebox{1.1}{$\otimes$}}_{\pi}\hspace{2pt}}
\newcommand{\tetriangle}{%
  \mathrel{\vbox{\offinterlineskip\ialign{%
    \hfil##\hfil\cr
    \scalebox{0.9}[0.3]{$\triangle$}\cr
    \scalebox{1.1}{$\otimes$}\cr
}}}}  
\newcommand{\pptp}{\tetriangle_{\absv{\pi}}} 
\newcommand{\bigtetriangle}{%
  \mathrel{\vbox{\offinterlineskip\ialign{%
    \hfil##\hfil\cr
    \scalebox{1.1}[0.3]{$\triangle$}\cr
    \scalebox{1}{$\bigotimes$}\cr
}}}}  
\DeclareMathOperator*{\bigpptp}{\bigtetriangle_{\absv{\pi}}} 
\newcommand{\sem}[1]{\llbracket #1\rrbracket}
\newcommand{\dg}{^\dagger}

\usepackage{tcolorbox,color}

\definecolor{cornellred}{RGB}{196,18,48}
\definecolor{dartmouthgreen}{RGB}{0,112,60}

\begin{document}

%
\title{Semantics of higher-order probabilistic programs with conditioning}

\author{\IEEEauthorblockN{Fredrik Dahlqvist}
\IEEEauthorblockA{University College London\\
f.dahlqvist@ucl.ac.uk}
\and
\IEEEauthorblockN{Dexter Kozen}
\IEEEauthorblockA{Cornell University\\
dexter.kozen@cornell.edu}
}

\maketitle

\pagestyle{plain}

\begin{abstract}
We present a denotational semantics for higher-order probabilistic programs in terms of linear operators between Banach spaces. Our semantics is rooted in the classical theory of Banach spaces and their tensor products, but bears similarities with the well-known semantics of higher-order programs \`{a} la Scott through the use \emph{ordered} Banach spaces which allow definitions in terms of fixed points.  Being based on a monoidal rather than cartesian closed structure, our semantics effectively treats randomness as a \emph{resource}. 
\end{abstract}


%
\IEEEpeerreviewmaketitle

\section{Introduction}

Probabilistic programming has enjoyed a recent resurgence of interest driven by new applications in machine learning and statistical analysis of large datasets. The emergence of probabilistic programming languages such as Church and Anglican, which allow statisticians to construct and sample distributions and perform Bayesian inference, has created a need for sound semantic foundations and tools for specification and reasoning. Several recent works have approached this task from various perspectives \cite{heunen2017convenient,scibior2017denotational,staton2017commutative,ehrhard2017measurable}.

One of the earliest works on the semantics of probabilistic programs was \cite{K81c}, in which operational and denotational semantics were given for an idealized first-order imperative language with random number generation. Programs and data were interpreted over ordered Banach spaces. Programs were modelled as positive and continuous linear operators on an ordered Banach space of measures. In \cite{K85a}, an equivalent predicate-transformer semantics was introduced based on ordered Banach spaces of measurable functions and shown to be dual to the measure-transformer semantics of \cite{K81c}.

In this paper revisit this approach. We identify a symmetric monoidal closed category $\Roban$ of ordered Banach spaces and regular maps that can serve as a foundation for higher-order probabilistic programming with sampling, conditioning, and Bayesian inference. Bayesian inference can be viewed as reversing the computation of a probabilistic program to infer information about a prior distribution from observations. We model Bayesian inference as computing the \emph{adjoint} of a linear operator and show how it corresponds to computing a \emph{disintegration}.

The extension to higher types is achieved through a tensor product construction in the category $\Roban$ that gives symmetric monoidal closure. Although not cartesian, the construction does admit an adjunction with homsets enriched with an ordered Banach space structure acting as internalized exponentials. To accommodate conditioning and Bayesian inference, we introduce `Bayesian types', in which values are decorated with a prior distribution. 
Based on this foundation, we give a type system and denotational semantics for an idealized higher-order probabilistic language with sampling, conditioning, and Bayesian inference. 

We believe our approach should appeal to computer scientists, as it is true to traditional Scott-style denotational semantics (see \cref{sec:Scott}). It should also appeal to mathematicians, statisticians and machine learning theorists, as it uses very familiar mathematical objects from those fields. For example, a traditional perspective is that a Markov process is just a positive linear operator of norm 1 between certain Banach lattices \cite[Ch. 19]{aliprantis}. These are precisely the morphisms of our semantics. Similarly, classical ergodic theory, which is key to proving the correctness of important algorithms like Gibbs sampling, is an important part of the theory of these operators \cite{eisner2015operator}. Our semantics therefore connects seamlessly with a wealth of results from functional analysis, ergodic theory, statistics, etc. We believe that this will greatly simplify the task of validating stochastic machine learning algorithms.

We should also mention that our semantics fits well with the view of entropy (randomness) as a computation resource, like time or space. True random number generators can only produce randomness at a limited rate; physically, randomness \emph{is} a resource \cite{hayes2001computing}. Our type system, being resource-sensitive, has some nice crypographical properties: it is forbidden by construction to use a sample more than once; that is, each operation consuming a random sample requires a fresh sample (component in a tensor product). 

\textit{Related works:} Two very powerful semantics for higher-order probabilistic programming have been recently developed in the literature. In \cite{heunen2017convenient,scibior2017denotational}, a semantics is given in terms of so-called quasi-Borel spaces. These form a Cartesian closed category and admit a notion of probability distribution and of a Giry-like monad of probability distributions. In \cite{ehrhard2017measurable} the authors develop a semantics in terms of measurable cones. These form a cpo-enriched Cartesian closed category which provides a semantics to a probabilistic extension of PCF that includes conditioning. The key differences with the present semantics are the following. 
First, these proposed mathematical universes come directly from the world of theoretical computer science, whilst as mentioned above, our semantics is rooted in the traditional mathematics of the objects being constructed by the programs. Second, quasi-Borel spaces and measurable cones form Cartesian closed categories, whereas we work in a monoidal closed category, with obvious implications in terms of resources (e.g. we cannot copy a value). Finally, our semantics of conditioning has been reduced to a mathematically very simple, but also very general construction (taking the adjoint of a linear operator, see \cref{sec:semobserve}), whilst in \cite{heunen2017convenient} un-normalized posteriors and normalization constants are computed pointwise, and \cite{ehrhard2017measurable} effectively hard-codes the rejection-sampling algorithm into the semantics.

The reader will find the proofs of most results in the Appendix, together with some background material on measure theory and tensor products of Banach spaces.
\section{Background}

We start by describing the mathematical landscape of our semantics. We assume that the reader is familiar with the basic definitions of measure theory and of what a (real) Banach space is (see \cite[Ch. 4, 6, 8-11]{aliprantis} for a gentle introduction in the spirit of this paper). 

\subsection{\textbf{Banach spaces, Disintegration and Bayesian inversion}}\label{subsec:Meas}
 
\subsubsection{\uline{Some important Banach spaces}}
Two classes of Banach spaces will appear repeatedly in this paper. 

First, for any measurable space $(X,\salg)$ we introduce the space $\Meas(X,\salg)$, or simply $\Meas X$, as the set of signed measures of bounded variation over $X$. $\Meas X$ is a Banach space: the linear structure is inherited pointwise from $\R$, and the norm is given by the total variation; see \cite[Th. 10.56]{aliprantis} for a proof that the space is complete. 

Second, for a measured space $(X,\salg,\mu)$ and $1\leq p<\infty$, the \emph{Lebesgue space} $\Lps(X,\mu)$ is the set of equivalence classes of $\mu$-almost everywhere equal $p$-integrable real-valued functions, that is to say functions $f: X\to\R$ such that
\[
\int \absv{f}^p~d\mu<\infty.
\]
The linear structure is inherited pointwise from $\R$ and the norm is given by $\norm{f}_p=\int \absv{f}^p~d\mu$. When $p=\infty$, the space $\Lps[\infty](X,\mu)$ is defined as the set of equivalence classes of $\mu$-almost everywhere equal bounded real-valued functions with the norm given by the essential supremum:
\[
\norm{f}_\infty=\inf \set{C\geq 0}{\absv{f(x)}\leq C\ \mu\text{-a.e.}}
\]
A proof that Lebesgue spaces are complete can be found in \cite[Th. 13.5]{aliprantis}. 

 \subsubsection{\uline{Disintegration}}\label{sec:disintegration}  Measurable spaces and maps form the category $\cMeas$. We define the functor $\mMeas:\cMeas\to\cMeas$ by setting $\mMeas X$ to be the set of signed measures of bounded variation on $X$ equipped with the smallest $\sigma$-algebra making all evaluation maps $ev_B: \mMeas X\to\R$, $\mu\mapsto ev_B(\mu)=\mu(B)$ measurable and by setting $\mMeas f:\mMeas X\to \mMeas Y, \mu\mapsto f_*(\mu)$, the pushforward measure of $\mu$, for any $f: X\to Y$ \footnote{This is just a generalisation of the Giry monad on $\cMeas$ \cite{giry1982categorical}. Note that $\Meas X$ and $\mMeas X$ share the same underlying set, but the former is a Banach space and the latter a measurable space.}. We define a \emph{measure kernel} to be a measurable map $f: X\to \mMeas Y$ such that for all $x\in X$ $\norm{f(x)}<K$ for some fixed $K\in\R^+$. A \emph{probability kernel} is a measure kernel such that $\norm{f(x)}=1$ and $f(x)\geq 0$ for all $x\in X$. A measure $\mu\in\mMeas X$ can also be pushed-forward through a measure kernel $f:X\to \mMeas Y$ to give a measure in $\mMeas Y$ via the definition
\begin{equation}\label{eq:kernelpushforward}
f_*(\mu)(B)=\int_X f(x)(B)~d\mu
\end{equation}
which converges since $f(x)(B)\leq f(x)(X)\leq \absv{f(x)}(X)=\norm{f(x)}<K$.

With these definitions in place we can introduce the important notion of \emph{disintegration} which underlies the semantics of Bayesian conditioning (see \cref{sec:semobserve}). We provide a slightly simplified version of the definition which will be enough for our purpose (see \cite[Def. 1]{chang1997conditioning} for a very general definition). Intuitively, given a measurable map $f: X\to Y$ and a finite measure $\mu$ on $X$, we say that \emph{$\mu$ has a disintegration w.r.t. $f$} if the fibres $f\inv(y)$ of $f$ can be equipped with measures $f\dg_\mu(y)$ which average out to $\mu$ over the pushforward measure $f_*(\mu)$. Formally, the disintegration of $\mu$ w.r.t. to $f$ is a measure kernel $f\dg_\mu: Y\to \mMeas X$ such that
\begin{itemize}
\item $f_*(f\dg_\mu(y))\propto\delta_y\text{ for }f_*(\mu)\text{-almost all }y\in Y$
\item $(f\dg_\mu)_*(f_*(\mu))=\mu $
\end{itemize}
In fact \cite[Th. 3]{chang1997conditioning} shows that $f\dg_\mu$ can be chosen to be a probability kernel. As can be seen from the first condition, a disintegration  -- if it exists at all -- is only defined up to a null set for the pushforward measure. For sufficiently well-behaved spaces, for example standard Borel spaces \cite[17.35]{kechris} or more generally metric spaces with Radon measures \cite[Th. 1]{chang1997conditioning}, disintegrations can be shown to always exist.

\subsubsection{\uline{Bayesian inversion}}\label{sec:BayesianInverse}The notion of disintegration is key to the understanding of Bayesian conditioning. The traditional setup is as follows: we are given a probability kernel $f: X\to \mMeas Y$ where $X$ is regarded as a parameter space and $f$ is regarded as a parametrized statistical model on $Y$, a space of observable values. We also start with a probability distribution $\mu$ on $X$ (the prior) which is regarded as the current state of belief of where the `true' parameters of the model lie. The problem is, given an observation $y\in Y$, to update the state of belief to a new distribution (the posterior) reflecting the observation. We must therefore find a kernel going in the opposite direction $f\dg_\mu: Y\to\mMeas X$. As shown in \cite{2017:FoSSaCS,MFPS2018} this reverse kernel can be built using a disintegration as follows. First we define a joint distribution $\gamma\in\mMeas(X\times Y)$ defined by 
\[
\gamma(A\times B)=\int_A f(x)(B)~d\mu,
\]
The Bayesian inverse $f\dg$, if it exists, is given by the probability kernel
\[
f\dg_\mu=(\pi_X)_* \circ  (\pi_Y)\dg_\gamma
\]
where $(\pi_Y)\dg_\gamma$ is the disintegration of the measure $\gamma$ along the projection $\pi_Y: X\times Y\to Y$ (it can be assumed to be a probability kernel). This construction clearly generalizes to all measure kernels.

\subsection{\textbf{Ordered Banach spaces}}\label{subsec:OrdBan}

\subsubsection{\uline{Regular Ordered Banach spaces}}\label{sec:robanobj}
An \emph{ordered vector space} $V$ is a vector space together with a partial order $\le$ which is compatible with the linear structure in the sense that for all $u,v,w\in V, \lambda\in \R^+$ 
\begin{align*}
u\hspace{-1pt}\le \hspace{-1pt}v\Rightarrow u+w\hspace{-2pt}\le \hspace{-2pt} v+w \qquad\text{and}\qquad u\hspace{-1pt}\le\hspace{-1pt} v\Rightarrow \lambda u\hspace{-1pt}\le \hspace{-1pt}\lambda v
\end{align*}
A vector $v$ in an ordered vector space $V$ is called \emph{positive} if $v\geq 0$ and the collection of all positive vectors is called the \emph{positive cone} of $V$ and denoted $V^+$. The positive cone is said to be \emph{generating} if $V=V^+ - V^+$, that is to say if every vector can be expressed as the difference of two positive vectors.

An \emph{ordered normed vector space} $V$ is an ordered vector space in which the positive cone $V^+$ is closed for the topology generated by the norm. A subset of the positive cone of particular importance will be the \emph{positive unit ball} $B^+(V)=\{v\geq 0 : \norm{v}\leq 1\}$. An \emph{ordered Banach space} is an ordered normed vector space which is complete. We can now describe the central class of object of this work: an ordered normed space is said to be \emph{regular} if it satisfies \cite[Ch. 9]{wong1973partially}:

\begin{enumerate}[label=R\arabic*]
\item \label{R1} if $-y\leq x\leq y$ then $\norm{x}\leq\norm{y}$
\item \label{R2} $\norm{x}=\inf\{\norm{y}:-y\leq x\leq y\}$
\end{enumerate} 
In particular, a regular ordered Banach space is an ordered Banach space which is regular. A few comments are in order. First note that if $-y\leq y$ then $0\leq 2y$, and thus $y$ is positive, so  \ref{R2} says that the norm of any vector can be approximated arbitrarily well by the norm of positive vectors. Note also that \ref{R2} implies that the positive cone is generating: for any $x\in V$, fix $\epsilon>0$, then by \ref{R2} there exists $y$ with $-y\leq x\leq y$ whose norm is $\epsilon$-close to that of $x$. Since
$
x=\frac{y+x}{2}-\frac{y-x}{2}
$, 
and since it follows from $-y\leq x\leq y$ that both $y+x$ and $y-x$ are positive, $x$ can indeed be expressed as the difference of two positive vectors. Regularity can be understood as the fact that the space is fully characterised by its positive unit ball \cite{min1983exponential}.

\subsubsection{\underline{Regular operators and $\Roban$}}\label{sec:regoperators}
As mentioned above, regular ordered Banach spaces are determined in a very strong sense by their positive cone which is generating and determines the norm (axiom R2). It is therefore natural to consider linear operators $f: U\to V$ between regular ordered Banach spaces which send positive vectors to positive vectors, i.e. such that $u\geq 0\Rightarrow f(u)\geq 0$. Such operators are called \emph{positive operators} and constitute a field of mathematical research in their own right \cite{zaanen2012introduction,aliprantis2006positive}. The collection $[U,V]_p$ of positive operators between two regular ordered Banach spaces clearly does not form a vector space, since it is not closed under scalar multiplication by negative reals. We therefore consider the span of this collection, that is to say the operators $f: U\to V$ which can be expressed as the difference between two positive operators, i.e. $f=h-g$ with $h,g\in [U,V]_p$. Such operators are called \emph{regular operators}, and we define the category $\Roban$ as \emph{the category whose objects are regular ordered Banach spaces and whose morphisms are regular operators}. Regular operators have the following important properties.

\begin{proposition}\label{prop:regbounded}
Regular operators on regular ordered Banach spaces are (norm) bounded.
\end{proposition}

\begin{theorem}[\cite{min1983exponential}]\label{Th:internalhom}
If $U,V$ are regular ordered Banach spaces and $[U,V]$ is equipped with the obvious linear structure, pointwise order and the \emph{regular norm}
\[
\norm{f}_r=\inf\left\{\norm{g}: -g\leq f\leq g \right\}
\]
where $\norm{g}=\sup\{\norm{g(u)} : \norm{u}\leq 1\}$ is the usual operator norm, then $[U,V]$ is a regular ordered Banach space. 
\end{theorem}
This result justifies the following notation: we will denote the regular ordered Banach space of operators between the regular ordered Banach spaces $U,V$ by $[U,V]$.

\subsubsection{\uline{Banach lattices}} \label{sec:banlattices} We now describe a particularly important class of regular ordered Banach spaces: the class of Banach lattices. Although this class of objects lacks the categorical closure properties that we seek (see \cref{subsec:ROBanCat}), most of the objects we will be dealing with are Banach lattices.

An ordered vector space $(V,\leq)$ is a \emph{Riesz space} if its partial order is a lattice. This allows the definition of the \emph{positive and negative part of a vector} $v\in V$ as
$v^+=v\vee 0, v^-=(-v)\vee 0$ and its \emph{modulus} as $\absv{v}=v\vee(-v)$. Note that $v=v^+ - v^-$, with $v^+,v^-$ positive, and the positive cone of a Riesz space is thus generating. A Riesz space is \emph{order complete} or \emph{Dedekind-complete} (resp. $\sigma$-order complete or $\sigma$-Dedekind complete) if every non-empty (resp. non-empty countable) subset of $V$ which is order bounded has a supremum\footnote{Order-completeness was called \emph{conditional completeness} in \cite{K81c}}.
A \emph{normed Riesz space} is a Riesz space equipped with a \emph{lattice norm}, i.e.\@ norm satisfying axiom R1 above.
A normed Riesz space is called a \emph{Banach lattice} if it is (norm-) complete. As stated, Banach lattices form a special class of regular ordered Banach spaces:

\begin{proposition}
Banach lattices are regular.
\end{proposition}

\begin{example}
Given a measurable space $(X,\salg)$, the space $\Meas(X,\salg)$ can be shown \cite[Th 10.56]{aliprantis} to be a Banach lattice. The Banach space structure was described above and the lattice structure is given by
\[
(\mu\vee \nu)(A)=\sup\{\mu(B)+\nu(A\setminus B)\mid B\text{ measurable }, B\subseteq A\}
\]
and dually for meets. The Hahn-Jordan decomposition theorem defines the positive and negative part of a measure in the Banach lattice $\Meas(X,\salg)$. 
\end{example}

\begin{example}\label{ex:LpSpaces1}
Given a measured space $(X,\salg,\mu)$ and $1\leq p\leq \infty$, the Lebesgue space $\Lps(X,\mu)$ is a Banach lattice with the pointwise order. In particular, for any $f\in \Lps(X,\mu)$, the positive and negative parts $f^+$ and $f^-$ of a function used in the definition of the Lebesgue integral defines the positive-negative decomposition of $f$ in the Banach lattice $\Lps(X,\mu)$.  We will say that $p,q\in \N\cup\{\infty\}$ are \emph{H\"{o}lder conjugate} if either of the following conditions hold: (i) $1<p,q<\infty$ and $\frac{1}{p}+\frac{1}{q}=1$, or (ii) $p=1$ and $q=\infty$, or (iii) $p=\infty$ and $q=1$.
\end{example}

The examples of Banach lattices described above are instances of an even better behaved class of objects called \emph{abstract Lebesgue spaces} or \emph{AL spaces}. They are defined by the following property of the norm: a Banach lattices $V$ is an AL space if for all $u,v\in V^+$
\begin{equation}\label{eq:AL}
\norm{u+v}=\norm{u}+\norm{v}\tag{AL}
\end{equation}
Not surprisingly, the Lebesgue spaces $\Lps[1](X,\mu)$ are examples of AL spaces, as are the Banach lattices $\Meas(X)$. 

\begin{theorem}[\cite{aliprantis2006positive}, Sec. 4.1]\label{Th:ALcomplete}
AL spaces are order-complete.
\end{theorem}

\subsubsection{\uline{Bands}}\label{sec:bands}
The order structure of Riesz spaces gives rise to classes of subspaces which are far richer than the traditional linear subspaces. An \emph{ideal} of a Riesz space $V$ is a linear subspace $U\subseteq V$ with the property that if $\absv{u}\leq \absv{v}$ and $v\in U$ then $u\in U$. An ideal $U$ is called a \emph{band} when for every subset $D\subseteq U$ if $\bigvee D$ exists in $V$, then it also belongs to $U$. Every band in a Banach lattice is itself a Banach lattice. Of particular importance in what follows will be the \emph{principal band generated by an element $v\in V$}, which we denote $V_v$  and can be described explicitly by
\[
V_v=\{w\in V\mid (\absv{w}\wedge n\absv{v}) \uparrow \absv{w}\}
\]
\begin{example}\label{ex:Band}\
Given a measure $\mu\in \Meas X$, the band $(\Meas X)_\mu$ generated by $\mu$ is the set of signed measures of bounded variation which are absolutely continuous w.r.t. $\mu$ \cite[Th.~10.61]{aliprantis}. 
The ordered version of the Radon-Nikodym theorem states that $(\Meas X)_\mu \simeq \Lps[1](X,\mu)$ as Banach lattices \cite[Th.~13.19]{aliprantis}. 
\end{example}

\subsubsection{\uline{K\"{o}the duals}}\label{sec:Kothe} There are two modes of `convergence' in an ordered Banach space: \emph{order convergence} and \emph{norm convergence}. The latter is well-known, the former less so. Let $D$ be a directed set, and let $\{v_\alpha\}_{\alpha\in D}$ be a net in an ordered Banach space $V$.  We say that $\{v_\alpha\}$ \emph{converges in order to $v$} if there exists a \emph{decreasing} net $\{u_\alpha\}_{\alpha\in D}$ with $\bigwedge u_\alpha = 0$, notation $u_\alpha\downarrow 0$, such that
\[
-u_\alpha\leq v_\alpha-v\leq u_\alpha \text{ for all }\alpha\in D
\]
If the directed set $D$ is $\N$ we get the notion of \emph{order-convergent sequence}.  Order and norm convergence of sequences are disjoint concepts, i.e.\@ neither implies the other (see \cite[Ex.~15.2]{zaanen2012introduction} for two counter-examples). However if a sequence converges both in order and in norm then the limits are the same (see \cite[Th.~15.4]{zaanen2012introduction}). Moreover, for \emph{monotone} sequences norm convergence implies order convergence \cite[Th.~15.3]{zaanen2012introduction}.

It is well known that bounded operators are continuous, i.e. \@ preserve norm-converging sequences. The corresponding order-convergence concept is defined as follows: an operator $T: V\to W$ between Riesz spaces is said to be \emph{$\sigma$-order continuous} if whenever $v_n\downarrow 0$, $Tv_n\downarrow 0$\footnote{Equivalently if $v_n\uparrow v$, i.e. $v_n$ is an increasing sequence with supremum $v$, implies $Tv_n \uparrow Tv$. Note the similarity with Scott-continuity, the only difference being the condition that sequences must be order-bounded.}. We can thus consider two types of dual spaces on an ordered Banach space $V$: on the one hand we can consider the \emph{norm-dual}:
\[
V^*=\{f: V\to \R : f\text{ is norm-continuous}\}
\]
and the \emph{$\sigma$-order-dual}:
\[
V\kd=\{f: V\to \R : f\text{ is $\sigma$-order continuous and regular}\}
\]
The latter is also known as the \emph{K\"{o}the dual} of $V$ \cite{1951:dieudonneKothe,zaanen2012introduction}.

\begin{theorem}\label{Th:Kothe}
The K\"{o}the dual $V\kd$ of a regular ordered Banach space $V$ is an order-complete Banach lattice.
\end{theorem}

\begin{example}
It is shown in e.g. \cite{zaanen2012introduction,ampba} that   
\begin{equation}
\Lps(X,\mu)\kd = \Lps[q](X,\mu)
\end{equation}
for \emph{any} H\"{o}lder conjugate pair  $1\leq p,q\leq \infty$. In particular the spaces $\Lps[1](X,\mu)$ and $\Lps[\infty](X,\mu)$ are K\"{o}the dual of each other. Note that they are \emph{not} ordinary duals.
\end{example}

\subsubsection{\uline{Categorical connections}.}\label{sec:catcon} We conclude this section by a summary of some results from \cite{MFPS2018} which provide a categorical connection between most of the topics covered so far. 

The category $\Krn$ is the category whose objects are pairs $(X,\mu)$ where $X$ is standard Borel spaces \cite{kechris} (in fact any class of measurable spaces for which disintegrations exist will do) and $\mu\in\mMeas X$. A morphism between $(X,\mu)$ and $(Y,\nu)$ is a measure kernel $f: X\to \mMeas Y$ such that $f_*(\mu)=\nu$ (where $f_*$ is defined in \eqref{eq:kernelpushforward}), in which case the morphism is denoted $f$ as well. As was shown in \cite{MFPS2018}, any two morphisms which disagree only on a null set can be identified, and the morphisms of $\Krn$ thus become \emph{equivalence classes} of almost everywhere equal measure kernels (see \cite{MFPS2018} for the technical details of this construction).

Now, we define some functors. First, as was shown in \cite{MFPS2018}, the  Bayesian inversion operation described in \cref{sec:BayesianInverse} defines a functor $(-)\dg: \Krn\to\Krn\op$ which leaves objects unchanged and sends a morphism $f:(X,\mu)\to(Y,\nu)$ to its Bayesian inverse $f\dg: (Y,\nu)\to (X,\mu)$ (we drop the subscript $\mu$ of $f\dg_\mu$ because it is made explicit from the typing). Note that $(f\dg)\dg=f$ \cite{MFPS2018}. We also define the functor $(-)\kd: \Roban\to\Roban\op$ which sends a regular ordered Banach space $X$ to its K\"{o}the dual, and a regular operator $T: U\to V$ to its adjoint $T\kd: V\kd\to U\kd$ defined in the usual way. Note that just as taking the K\"{o}the dual gives an order-complete space, the adjoint $T\kd$ of a regular operator is an order-continuous regular operator \cite[Ch. 26]{zaanen2012introduction}.

Connecting the categories, we define for each $1\leq p\leq \infty$ the functor $\Lps: \Krn\to \Roban\op$ which sends a $\Krn$-object $(X,\mu)$ to the Lebesgue space $\Lps(X,\mu)$ and a $\Krn$-arrow $f:(X,\mu)\to(Y,\nu)$ to the operator $\Lps f: \Lps(Y,\nu)\to\Lps(X,\mu)$, $\phi\mapsto\lambda x~.~\int_Y \phi ~df(x)$. We also define the functor $\Meas_{-}: \Krn\to\Roban$ which sends an object $(X,\mu)$ to the band $(\Meas X)_\mu$ and a morphism $f: (X,\mu)\to (Y,\nu)$ to the operator $\Meas f: (\Meas X)_\mu\to(\Meas Y)_\nu$, $\rho\mapsto \lambda B~.~ \int_{X} f(x)(B) ~d\rho$.

The functors $\Meas_{-}, \Lps[1]\circ(-)\dg$ and $(-)\kd\circ\Lps[\infty]$ of type $\Krn\to\Roban$ are related via natural transformations which play a major role in measure theory \cite{MFPS2018}:
\newcommand{\RN}{\mathrm{RN}}
\newcommand{\MR}{\mathrm{MR}}
\newcommand{\FR}{\mathrm{FR}}
\newcommand{\RR}{\mathrm{RR}}
\begin{itemize}
\item $\RN: \Meas_{-}\to \Lps[1]\circ(-)\dg$ acts at $(X,\mu)$ by sending a measure $\nu\ll\mu$ to its Radon-Nikodym derivative $\frac{d\mu}{d\nu}$.
\item $\MR: \Lps[1]\circ(-)\dg\to \Meas_{-}$ acts at $(X,\mu)$ by sending an $\Lps[1]$-map $f$ to its Measure Representation $f\mu$.
\item $\FR: \Meas_{-}\to (-)\kd\circ\Lps[\infty]$ acts at $(X,\mu)$ by sending a measure $\mu$ to its Functional Representation $\lambda\phi. \int \phi~d\mu$.
\item $\RR:  (-)\kd\circ\Lps[\infty]\to \Meas_{-}$ acts at $(X,\mu)$ by sending an $\Lps[\infty]$-functional $F$ to its Riesz Representation $\lambda B. F(1_B)$. 
\end{itemize}
The natural transformations $\RN$ and $\MR$ are inverse of each other, as are $\FR$ and $\RR$, proving natural isomorphisms between the three functors. We summarize these relationships in the following diagram:
\begin{equation}\label{diag:summary}
\xymatrix@R=8ex@C=10ex
{
 & \Roban\\
\Roban\op\ar[ur]^{(-)^\sigma} \ar@<-1ex>@{=>}[r]_-{\RR}  &  \ar@<-1ex>@{=>}[r]_-{\RN} \ar@<-2pt>@{=>}[l]_-{\FR}  & \Krn\op\ar[ul]_{\Lps[1]}\ar@<-2pt>@{=>}[l]_-{\MR}  \\
 &\Krn\ar[ur]_{(-)\dg}\ar[ul]^{\Lps[\infty]}\ar[uu]^>>>>>>>>{\Meas_{-}}
}
\end{equation}

\subsection{\textbf{Tensor products of ordered Banach spaces}}\label{subsec:tensor}

We start by describing the tensor product of vector spaces from the perspective of computer science. We will then discuss how the tensor product can be normed and ordered.

\subsubsection{\uline{Introduction to the tensor product}}\label{sec:tpVec}
As was already highlighted in \cite{K81c} in the case of probabilistic programming, and subsequently in the development of semantics for quantum programming languages (e.g. \cite{selinger2008fully}), it may be desirable to interpret programs as linear operators in a category of vector spaces. Indeed, this is precisely what this paper advocates for probabilistic programming languages. However, a difficulty quickly emerges if one wants to include higher-order features. Consider a map in two arguments $f: U\times V\to W$. The most basic facility provided by higher-order reasoning is the ability to \emph{curry} such a map and define the two curried maps
\[
\hat{f}: U\to W^V \qquad\text{and}\qquad \tilde{f}: V\to W^U
\]
by fixing one argument or the other. Since we want both curried map $\hat{f}$ and $\tilde{f}$ to be linear, it is easy to see that $f$ must be linear in each arguments separately, in particular 
\[
f(\lambda u,v)=\lambda f(u,v) \quad\text{and}\quad f(u,\lambda v)=\lambda f(u,v)=\lambda(u,v)
\]
Such a map is referred to as a \emph{bilinear} map, it is linear in each argument \emph{separately}. However being bilinear is incompatible with being linear: by definition of the product linear structure if $f$ were also linear we would have
\[
\lambda f(u,v)=f(\lambda(u,v))=f(\lambda u,\lambda v)=\lambda f(u,\lambda v)=\lambda^2 f(u,v)
\]
which is clearly a contradiction if $\lambda\neq 1$. Thus $f$ is not a valid morphism if we want our semantic universe to consist of linear maps between vector spaces. Fortunately, for any pair of vector spaces $U,V$ there exists a special object, the tensor product $U\tp V$, which \emph{linearizes} bilinear maps, i.e. such that any bilinear map $f: U\times V\to W$ corresponds to unique \emph{linear map} $\bar{f}:U\tp V\to W$ (and vice-versa). This can be phrased in terms of a universal property: there exists a universal bilinear map $\tp: U\times V\to U\tp V$, $(u,v)\mapsto u\tp v$ such that for any bilinear map $f: U\times V\to W$ there exists a unique linear map $\bar{f}:U\tp V\to W$  making the following diagram commute:
\begin{equation}\label{diag:tensorUP}
\xymatrix
{
U\times V\ar[r]_{\tp}\ar[d]_{f} & U\tp V\ar@{-->}[dl]^{\bar{f}} \\
W
}
\end{equation}
The tensor product can be built explicitly as follows: it is the free vector space over $U\times V$ quotiented by the following identities:
\begin{align}
&(u+u',v)=(u,v)+(u',v), \quad (u,v+v')=(u,v)+(u,v'), \nonumber \\
&(\lambda u,v)=(u,\lambda v)=\lambda(u,v) \label{eq:tensorconstruction}
\end{align}
It is not too hard to see that the last identity is precisely what is needed to fix the contradiction $\lambda f(u,v)=\lambda^2 f(u,v)$ described above and reconcile currying with linearity. By definition, an element $x\in U\tp V$ will be a (finite) linear combination of equivalence classes of the generators $(u,v)$ -- denoted $u\tp v$ -- under the identities \eqref{eq:tensorconstruction}, formally $x=\sum_i u_i\tp v_i$.

\subsubsection{\uline{Tensor product of Banach spaces}}\label{sec:tpBan} Suppose now that both $U$ and $V$ are Banach spaces, in particular that they carry a norm. How do we define a norm on $U\tp V$, and how do we ensure that the space is complete for this norm? For Hilbert spaces there is a straightforward construction, but since we shall be dealing with Banach spaces which are not Hilbert spaces, we will require the much more subtle theory of tensor products of Banach spaces originally developed by Grothendieck \cite{grothendieck1955tenseurs}. We refer to \cite{ryan2013introduction} for a good introduction. 

The initial difficulty with the construction of a norm is that by definition of the tensor product, each vector has many representations. For example, the vector $(2,2)\otimes (3,3)\in \R^2\tp\R^2$, i.e. the equivalence class of the pair $((2,2),(3,3))$ under the equations \eqref{eq:tensorconstruction}, can also be expressed as $(1,1)\otimes (6,6)$, and these representations are built from vectors with very different norms.  Assuming that we want the norm of the tensor be be defined from the norms of its components, which representation do we choose? There is no unique solution to this question, but Grothendieck proposed one extremal solution by defining for any $x\in U\tp V$
\begin{equation}\label{eq:projnorm}
\norm{x}_\pi=\inf\left\{\sum_{i=1}^n \norm{u_i}\norm{v_i} :  x=\sum_{i=1}^n u_i\otimes v_i\right\}
\end{equation}
This definition of $\norm{\cdot}_\pi$ defines a norm \cite[Prop 2.1]{ryan2013introduction} which is called the \emph{projective norm}. However, the space $U\tp V$ equipped with the projective norm is in general \emph{not complete}. One therefore \emph{defines} the \emph{projective tensor product} of two Banach spaces $U,V$ as the \emph{completion} of $U\tp V$ under the projective norm, i.e. the space of equivalence classes of Cauchy sequences in $U\tp V$ converging to the same point. This space will be denoted $U\ptp V$ and one can describe the projective norm of elements in this space as follows:
\[
\norm{x}_\pi\hspace{-2pt}=\hspace{-1pt}\inf\hspace{-2pt}\left\{\sum_{i=1}^\infty \norm{u_i}\norm{v_i} :  \sum_{i=1}^\infty \norm{u_i}\norm{v_i}\hspace{-1pt}<\hspace{-2pt}\infty, x\hspace{-2pt}=\hspace{-2pt}\sum_{i=1}^\infty\hspace{-1pt} u_i\hspace{-1pt}\otimes\hspace{-1pt} v_i\hspace{-1pt}\right\}
\]
In \cref{sec:tpVec} we saw how the tensor product can be used as a way to linearize bilinear maps. This property extends naturally to the normed case, and it can be shown that the projective tensor product $U\ptp V$ linearizes \emph{bounded bilinear maps} \cite[Th 2.9]{ryan2013introduction} in the sense that there exists a universal bounded bilinear maps $U\times V\to U\ptp V$ satisfying  the universal property of \eqref{diag:tensorUP} w.r.t bounded bilinear maps.

The projective tensor product of two Banach spaces is in general fairly inscrutable. However, one can explicitly describe projective tensor products involving important objects for our semantics. When one component is an $\Lps[1]$-space we have:

\begin{theorem}[Radon-Nikodym and \cite{ryan2013introduction} p. 43]\label{Th:pteL1L1}For finite measures $\mu,\nu$ on measurable spaces $X,Y$ respectively,
$(\Meas X)_\mu\ptp (\Meas Y)_\nu\simeq (\Meas(X\times Y))_{\mu\times\nu}$. 
\end{theorem}

The operator $\times:\Meas X\times\Meas Y\to \Meas(X\times Y)$ taking the product of measures is bilinear. Therefore there exists a unique map $\Meas X\ptp \Meas Y\to\Meas(X\times Y)$ mapping any tensor $\mu\ptp\nu$ to the product measure. In this sense, $\Meas{X}\ptp \Meas{Y}$ is the subspace of $\Meas{(X\times Y)}$ which is generated by taking linear combinations of product measures, and then closing under the projective norm.

\begin{theorem}\label{Th:pte:MM}
The projective tensor product $\Meas{X}\ptp \Meas{Y}$ is isometrically embedded in $\Meas{(X\times Y)}$. 
\end{theorem}

\subsubsection{\uline{Tensor product of ordered Banach spaces}}\label{sec:tpRoBan}

We conclude this brief description of tensor products by examining the case of interest to us, namely \emph{regular ordered} Banach spaces. In the most important examples the construction is isomorphic as Banach spaces to the unordered case, and we will therefore not dwell too long on the theory of tensor product of ordered Banach spaces developed in \cite{fremlin1972tensor,fremlin1974tensor,wittstock1974ordered}. The main idea of the definition is to reflect the central role of \emph{positive} vectors in the theory of regular ordered Banach spaces, and in particular the fact that the positive cone is generating and determines the norm (axiom R1, R2 above). The same should hold for any ordered tensor product.

Given two ordered regular Banach spaces $U,V$, their tensor product $U\otimes V$ is equipped with the \emph{positive projective norm} $\norm{\cdot}_{\absv{\pi}}$ defined as
\begin{align*}
\norm{x}_{\absv{\pi}}=\inf\left\{\sum_{i=1}^n \norm{u_i}\norm{v_i} :  u_i\in U^+, v_i\in V^+, \right. \\
\left. -\sum_{i=1}^n u_i\otimes v_i\leq x\leq \sum_{i=1}^n u_i\otimes v_i\right\}
\end{align*}
Note the similarity with \eqref{eq:projnorm}, and the role played by positive vectors in this definition. As in the unordered case $U\otimes V$ is not complete for the positive projective norm, and we must therefore take its completion which we call the \emph{positive projective tensor} of $U$ and $V$ and denote by $U\pptp V$.

Since $\Lps[1]$-spaces and $\Meas(X,\salg)$-spaces are examples of AL-spaces, the following result shows that we can in practice often ignore the subtleties of the positive projective tensor product and rely on the descriptions of the ordinary projective tensor products.
\begin{theorem}[\cite{fremlin1974tensor}, Th. 2B]\label{Th:ppte:ALspaces}
If $E$ is an $AL$-space and $F$ is any regular ordered space, then $E\pptp F$ and $E\ptp F$ are isomorphic as Banach spaces.
\end{theorem}

\subsection{\textbf{The closed monoidal structue of $\Roban$}}\label{subsec:ROBanCat}

\subsubsection{\underline{Tensor product of regular operators}} 
In \cref{sec:tpVec} and \cref{sec:tpBan} we saw how the tensor and projective tensor products can be used to linearize bilinear and bounded bilinear maps respectively. The positive projective tensor product fulfils the same role for \emph{positive} (and thus bounded by Prop. \ref{prop:regbounded}) bilinear maps:
there exists a universal positive bilinear map $U\times V\to U\pptp V$ satisfying the universal property of \eqref{diag:tensorUP} w.r.t positive bilinear maps \cite[2.7]{wittstock1974ordered}. This universal property of tensor products provides a definition of $\pptp$ as a bifunctor on $\Roban$. Let $f: U\to X, g: V\to Y$ be positive operators, then the map
\[
\pptp \circ~ (f\times g): U\times V\to X\times Y\to X\pptp Y
\]
is positive and bilinear, and thus there exists a unique positive operator $U\pptp V\to X\pptp Y$ which is denoted $f\pptp g$. This provides the definition of the bifunctor $\pptp$ on morphisms.

\subsubsection{\underline{The closed monoidal structure}}

As we saw in Th. \ref{Th:internalhom}, the category $\Roban$ has internal homs, and these interact correctly with positive projective tensor products.

\begin{theorem}[\cite{min1983exponential}]
For every regular ordered Banach space $U$, the tensoring and homming operations $-\pptp U$ and $[U,-]$ define functors $\Roban\to\Roban$ such that 
\[
-\pptp U \dashv [U,-]
\]
\end{theorem}

The positive projective tensor defines a symmetric monoidal structure on $\Roban$ with $\R$ as its unit -- since $U\tp \R\simeq U\simeq \R\tp U$ at the level of the underlying vector spaces --  and the obvious isomorphisms $U\pptp V\to V\pptp U$ inherited from the isomorphism $U\tp V\to V\tp U $ between the tensor product of the underlying vector spaces.
The category $\Roban$ is thus \emph{symmetric monoidal closed}.

\section{A higher-order language with conditioning}\label{sec:syntax}

\subsection{\textbf{A type system}}\label{sec:types}
We start by defining a type system for our language. Our aims are to (a) have enough types to write some realistic programs for example including multivariate normal or chi-squared distributions, (b) have higher-order types, (c) provide special types for Bayesian learning: \emph{Bayesian types}.

Our type grammar is given as follows:
\begin{align}
\Type::= & m \mid \tInt^n \mid \tReal^n \mid \tPosDef \mid \nonumber \\
& (\Type,\mu) \mid  \Type\otimes\Type \mid \Type \to \Type \mid \tMeas\Type \label{eq:types}
\end{align}
where $1\leq m,n\in\N$ and $\mu:\Type$. We will refer to $m$, $\tInt^n$, $\tReal^n,$ $\tPosDef$ as \emph{ground types}. As their name suggest they are to be regarded as the types of (possibly random) elements of finite sets, vectors of integers, vectors of reals and positive semi-definite matrices, that is to say \emph{covariance matrices}. We will write $\tInt^1$ and $\tReal^1$ as $\tInt$ and $\tReal$. This is by no means an exhaustive set of ground types, but sufficiently rich to consider some realistic probabilistic programs. The type $1\in \N$ will be referred to as the \emph{unit type} and denoted $\tUnit$ and the type $2\in \N$ will be referred to as the boolean type and denoted $\tBool$.

The type constructors are the following. First, given a term $\mu$ of type $\Type$, we can build the pointed type $(\Type,\mu)$. We will call these types \emph{Bayesian types} because $\mu:\Type$ will be interpreted as a \emph{prior}. Bayesian types will support conditioning and thus Bayesian learning. As we shall see, our Bayesian types also fulfil a role in the semantics of variable assignment. As is the tacit practise in Anglican, we will consider that assigning a (possibly random) value to a variable is equivalent to assigning a prior to the type of this variable. For example the program $\tt{x:=2.5}$ which assigns the value $2.5$ to the variable $\tt{x}$ can be understood as placing a (deterministic) prior on the reals, namely $\delta_{2.5}$. Similarly, the program $\tt{x:=sample(normal(0,1))}$ which assigns to $\tt{x}$ a value randomly sampled from the normal distribution $\mathcal{N}(0,1)$ with mean 0 and standard deviation 1 can be understood as setting the prior $\mathcal{N}(0,1)$ on the reals. In a slogan:
\begin{center}
\textit{Bayesian types = Assigned types}
\end{center}
Note however that this slogan is only valid for assignments \emph{without free variables}, indeed a prior cannot be parametric in some variables, it represents definite information. This caveat will be reflected in the type system.

We then have two binary type constructors: \emph{tensor types} and \emph{functions types} which together will support higher-order reasoning. Finally, we have a unary type constructor used to define higher-order probabilities.

We isolate the following two sub-grammars of types whose semantic properties will be essential to the typing of certain operations. First we define \emph{order-complete types} as the types generated by the grammar 
\begin{align}
\Type[S]&::=\Type[G]\mid  (\Type[G],\mu) \mid (\Type[G],\mu) \otimes(\Type[G],\mu)\qquad\Type[G]\text{ in ground types}\\
\Type&::=\Type[S]\mid \Type[S]\to\Type[S] \mid \tMeas\Type[S]\label{eq:completetypes}
\end{align}
Second, we define \emph{measure types} as the types generated by all the constructors of grammar \eqref{eq:types} \emph{except the function-type constructor}.

\begin{remark}
We could easily add product types to our type system since the category in which we interpret types ($\Roban$) is complete, but we feel that this would distract from the central role played by the tensor product. It would also introduce the possibility of copying, which from the perspective of randomness as a resource is problematic. This is why we have `hard-coded' the products which we do need (tuples of integers and reals) as ground types.
\end{remark}

\subsubsection*{\uline{Subtyping relation}} We will need to formalise the fact that a Bayesian type $(\Type,\mu)$ is a subtype of the type $\Type$. For this we introduce a subtyping relation denoted $<:$ freely generated by the rules

\begin{minipage}{.45\textwidth}
\footnotesize
\begin{tabular}{l l l}
\AxiomC{}
\UnaryInfC{$(\Type,\mu)<:\Type $}
\DisplayProof &

\AxiomC{$\Type[S]<:\Type[S']$}
\AxiomC{$\Type<:\Type'$}
\BinaryInfC{$\Type[S]\otimes\Type<:\Type[S']\otimes\Type'$}
\DisplayProof  & 

\AxiomC{$\Type[S']<:\Type[S]$}
\AxiomC{$\Type<:\Type'$}
\BinaryInfC{$\Type[S]\to \Type<:\Type[S']\to\Type'$}
\DisplayProof 
\end{tabular}
\end{minipage}

\vspace{4pt}
As will become clear when we define the semantics of types, we can also use the subtyping relation to add information about the absolute continuity of one built-in measure w.r.t another in the type system. For example, since a beta distribution is absolutely continuous w.r.t. to a normal distribution, if $e_1$ is a program constructing a beta distribution and $e_2$ is a program constructing a normal distribution, we could add the rule 
\footnotesize
\AxiomC{}
\UnaryInfC{$(\tReal,e_1)<:(\tReal,e_2) $}
\DisplayProof .
\normalsize

\subsubsection*{\uline{Contexts}} are maps $\tGamma:\N\to \tt{Types}$ -- the free algebra of all types generated by \eqref{eq:types}  -- which send cofinitely many integers to the unit type $\tt{unit}$. We will write $\supp(\Gamma)$ for the set $\{i\mid \Gamma(i)\neq \tUnit\}$ and use the traditional notation $\tGamma[i\tt{\mapsto \Type}]$
to denote the context mapping $i$ to $\Type$ and all $j\neq i$ to $\tGamma(j)$. To each context we can associate the finite tensor type $\bigotimes_{i=1}^n \tGamma(i)$ where $n=\sup\supp(\tGamma)$. When a context $\tGamma$ appears to the \emph{right} of a turnstile in the typing rules which follow we will implicitly perform this conversion from context to type. 

Our contexts are a dynamic version of the static context of \cite{K81c} which consists of a constant map on $\N$ to a single type. They are in some respects similar to the heap models of separation logic, and for notational clarity we will require similar operations on contexts as on heaps: a notion of compatibility, of union and of difference. Given two contexts $\tGamma_1,\tGamma_2$ we will say that they are \emph{compatible} if $\tGamma_1(i)=\tGamma_2(i)$ for all $i\in \supp(\tGamma_1)\cap\supp(\tGamma_2)$, and we will then write $\tGamma_i\pitchfork\tGamma_2$. For any two compatible contexts $\tGamma_1\pitchfork\tGamma_2$ we define the \emph{union context} $\tGamma_1\oplus\tGamma_2$ as the union of their graphs, which is a function by the compatibility assumption. We define the \emph{difference context} $\tGamma_1\ominus\tGamma_2$ as the map sending $i \mapsto \tGamma_1(i)$ if $i\notin\supp(\tGamma_2)$ and to $\tUnit$ otherwise. In particular $\tGamma_1\ominus\tGamma_2=\tGamma_1$ if the supports are disjoint.

\subsection{\textbf{Syntax}}

We define an ML-like language allowing imperative features like variable assignments, conditionals and while loops within a functional language.
 
\subsubsection{\uline{Expressions}} \hspace{1pt}

\begin{minipage}{.45\textwidth}
\footnotesize
\begin{align*}
e::=~ & n\in\N^k \mid r\in\R^k\mid  m\in\PosDef(n) \mid   & \text{Constants}\\
& \tt{op}(e,\ldots,e)\mid & \text{Built-in operations}\\
& x_i \mid &i\in\N,\text{ Variables }\\
& x_i:=e\mid &\text{Assignment}\\
& e ; e \mid &\text{Sequential composition}\\
& \tt{let~} x_i=e\tt{~in~}e & \text{Sequencing}\\
& \tt{fn ~}x_i~.~ e \mid &\lambda\text{-abstraction}\\
& e(e)\mid &\text{Function application}\\
& \tt{if}\quad e \quad \tt{then} \quad e\quad\tt{else}\quad e\mid &\text{Conditional}\\
& \tt{while}\quad e \quad \tt{do}\quad e \mid& \text{Iterations}\\
& \tt{sample}(e) \mid &\text{Sampling}\\
& \tt{sampler}(e) \mid &\hspace{-3em}\text{Packages a program as a sampler}\\
& \tt{observe}(e) & \text{Conditioning}
\end{align*}
\end{minipage}

\vspace{5pt}
Every built-in operation must come equipped with typing instruction which we will write as an $n+1$-tuple $(\Type[S_1],\ldots,\Type[S_n],\Type)$ where the first $n$ components are ground types specifying the input types and the last component is a ground type or measures over a ground type specifying the output type. For example the boolean connective $\texttt{or}$ would come with typing $(\tBool,\tBool,\tBool)$, the sine function $\texttt{sin}$ with typing $(\tReal,\tReal)$ and the function $\tt{normal}$ constructing a normal distribution would come with typing $(\tReal ,\tPosDef[1], \tMeas~\tReal)$, where the first input is the mean, the second is the standard deviation (a 1-dimensional covariance matrix, i.e. a positive real) and the output is a measure over the reals.

\subsubsection{\uline{Well-typed expressions}} 
The typing rules for our language are gathered in Fig. 1.  We will discuss these rules in detail when we define the denotational semantics of our language in \cref{sec:semantics}, but we can already make some observations. 

It is important to realize that memory-manipulating rules in effect have a sequent on the \emph{right} of the turnstile, formally represented by an integer-indexed tensor product type (see \cref{sec:types}), whilst the other rules just have a type. Syntactically and semantically however, we make no distinction between these two cases. A useful way to think about our system is as follows: a purely functional computation $\context e:\Type$ will consume a context $\tGamma$ and output a value of type $\Type$. A program with imperative features $\context e:\tt{\Delta}$ modifying an internal store $\tGamma$ to a new store $\tt{\Delta}$ should be thought of as consuming a context $\tGamma$ and outputting the totality of its internal state $\tt{\Delta}$.

The only way to explicitly create a Bayesian type is through a variable assignment without free variables: a prior must contain definite information, not information which is parametric in variables. Only `measure types' can form Bayesian types. 


The sequential composition rule looks daunting, but it is simply a version of the cut rule with a bit of bookkeeping to make sure contexts do not conflict with one another. 

Note finally that our \texttt{observe} statement applies to a term of type $\Type$, intuitively we observe a possibly random element of type $\Type$. This is slightly different from the syntax of \texttt{observe} in Anglican where a \emph{distribution} is observed. Semantically, the difference disappears since a possibly random element is modelled by a distribution.

\begin{figure*}[h!]\label{fig:typingrules}
\begin{minipage}{\textwidth}
\footnotesize

\begin{tabular}{ p{3.4cm}  l   l }
\hline
\\
\textbf{Constants:} &

\AxiomC{}
\RightLabel{$i\in n$}
\UnaryInfC{$\context[\emptyset] i:n$}
\DisplayProof & 

\AxiomC{}
\RightLabel{$n\in\N^k$}
\UnaryInfC{$\context[\emptyset] n:\tInt^k$}
\DisplayProof 

\\ \\

&

\AxiomC{}
\RightLabel{$r\in\R^k$}
\UnaryInfC{$\context[\emptyset] r:\tReal^k$}
\DisplayProof &

\AxiomC{}
\RightLabel{$M\in\PosDef(n)$}
\UnaryInfC{$\context[\emptyset] M:\tPosDef$}
\DisplayProof

\\ \\

\textbf{Variables and subtyping:} &

\AxiomC{}
\UnaryInfC{$[i\mapsto \Type] \vdash x_i: \Type$}
\DisplayProof &

\AxiomC{$\tt{\Gamma}[i\mapsto\Type[S]]\vdash e:\tt{\Delta}[j\mapsto\Type]$}
\RightLabel{$\Type[S]'<:\Type[S], \Type<:\Type'$}
\UnaryInfC{$\tt{\Gamma}[i\mapsto\Type[S]']\vdash e:\tt{\Delta}[j\mapsto\Type']$}
\DisplayProof 

\\ \\

& 
\AxiomC{$\context[\emptyset] e:\Type$}
\RightLabel{$\Type$ measure type}
\UnaryInfC{$\context[\emptyset] e: (\Type,e)$}
\DisplayProof
%
%
%
%

\\ \\

\textbf{Built-in operations:} &

\multicolumn{2}{l}{
\AxiomC{$\context[\Gamma_1] e_1:\tt{S_1}$}
\AxiomC{$\cdots$}
\AxiomC{$\context[\Gamma_n] e_n:\tt{S_n}$}
\RightLabel{$\tt{op}\text{ of type }(\tt{S_1,\ldots,S_n,T}), \supp(\tGamma_i)\cap \supp(\tGamma_j)=\emptyset, i\neq j$}
\TrinaryInfC{$\context[\Gamma_1,\ldots,\Gamma_n] \tt{op}(e_1,\ldots,e_n):\Type$}
\DisplayProof
}

\\ \\

\textbf{Assignment} &

\AxiomC{$\context e:\Type $}
\UnaryInfC{$\tt{\Gamma[i\mapsto\Type]}\vdash x_i:=e :\Type$}
\DisplayProof &

\AxiomC{$\context[\emptyset] e:\Type $}
\RightLabel{$\Type$ measure type}
\UnaryInfC{$\tt{[i\mapsto\Type]}\vdash x_i:=e :(\Type,e)$}
\DisplayProof 

\\ \\

\textbf{Sequencing:} &
\multicolumn{2}{l}{
\AxiomC{$\context e_1:\Type[S]$}
\AxiomC{$\tt{\Delta}[i\mapsto \Type[S]] \vdash e_2:\Type$}
\RightLabel{$\supp(\tGamma)\cap\supp(\tt{\Delta})=\emptyset$}
\BinaryInfC{$\context[\Gamma\oplus\Delta] \tt{let~}x_i=e_1\tt{~in~} e_2 :\Type$}
\DisplayProof
}

\\ \\

\textbf{Sequential composition:} & 

\multicolumn{2}{l}{
\AxiomC{$\context[\Gamma_1] e_1:\tt{\Delta_1}$}
\AxiomC{$\context[\Gamma_2] e_2:\tt{\Delta_2}$}
\RightLabel{$\tt{\Delta_1\pitchfork \Gamma_2,~\Gamma_1\pitchfork (\Gamma_2\ominus\Delta_1),~ \Delta_2\pitchfork (\Delta_1\ominus\Gamma_2)}$}
\BinaryInfC{$\context[\Gamma_1\oplus(\Gamma_2\ominus\Delta_1)] e_1;e_2 :\tt{(\Delta_1\ominus\Gamma_2)\oplus\Delta_2}$}
\DisplayProof 
}

\\ \\

\textbf{$\lambda$-abstraction and function application:} & 


\AxiomC{$\tt{\Gamma[i\mapsto S]}\vdash e:\Type$}
\RightLabel{$i\notin\supp(\tGamma)$}
\UnaryInfC{$\context\tt{fn}\hspace{2pt} x_i~.~e: \tt{S}\to \Type$}
\DisplayProof  &

\AxiomC{$\context e_1: \tt{S}$}
\AxiomC{$\context[\Delta] e_2 : \tt{S}\to\Type$}
\BinaryInfC{$\context[\Gamma,\Delta] e_2(e_1):\Type$}
\DisplayProof
 
\\ \\
\textbf{Imperative control flow:} & 

\multicolumn{2}{l}{
\AxiomC{$\context e_1:\tBool$}
\AxiomC{$\context e_2:\Type$}
\AxiomC{$\context e_3:\Type$}
\RightLabel{$\Gamma$ order-complete type}
\TrinaryInfC{$\context \tt{if} \enskip e_1 \enskip \tt{then}\enskip e_2\enskip\tt{else}\enskip e_3:\Type$}
\DisplayProof 
}

\\ \\

&

\AxiomC{$\context e_1:\tBool$}
\AxiomC{$\context e_2:\tt{\Gamma}$}
\RightLabel{$\tt{\Gamma}$ order-complete type}
\BinaryInfC{$\context\tt{while}\enskip e_1\enskip \tt{do}\enskip e_2:\tt{\Gamma}$}
\DisplayProof

\\ \\

\textbf{Probabilistic operations:}   & 

\AxiomC{$\context e: \Type$}
\RightLabel{$\Type$ measure type}
\UnaryInfC{$\context \tt{sampler}(e):\tMeas\Type$}
\DisplayProof  

& 
\AxiomC{$\context e:\tMeas \Type$}
\RightLabel{$\Type$ measure type}
\UnaryInfC{$\context\tt{sample}(e):\Type$} 
\DisplayProof 


 \\ \\

&  

\multicolumn{2}{l}{
\AxiomC{$[i\mapsto(\mathtt{S},\mu)]\vdash e: \Type$}
\RightLabel{$\Type[S],\Type$ measure type, $\Type[S]$ order-complete type}
\UnaryInfC{$[i\mapsto(\mathtt{S},\mu)]\vdash \tt{observe}(e): (\Type, e[x_i/ \mu])\to(\Type[S],\mu)$}
\DisplayProof
}
\\ \\
\hline
\end{tabular}
\end{minipage}
\caption{Typing rules}
\end{figure*}

\subsubsection{\uline{A simple example}}\label{sec:gaussianex}
It is not hard (but notationally cumbersome) to type-check the following simple Gaussian inference program against the inference rules of Fig. 1.
\begin{lstlisting}[frame=single,mathescape]
let x=sample(normal(0,1)) in
observe(sample(normal(x,1)))
\end{lstlisting}
In the empty context the program above evaluates to a function of type
\begin{align}
& \mathtt{(real,sample(normal(sample(normal(0,1)),1)))} \nonumber \\
& \to \mathtt{(real, sample(normal(0,1)))}\label{eq:gaussianinf}
\end{align}
which, as we will see in \cref{sec:semantics}, is what we want semantically.

\begin{comment}
\begin{itemize}
\item Properties of the type system?
\end{itemize}
\end{comment}

\section{Denotational semantics}\label{sec:semantics}
As the reader will have guessed we will now provide a denotational semantics for the language described in \cref{sec:syntax} in the category $\Roban$ of regular ordered Banach spaces.

\subsection{\textbf{Semantics of types}}

For \emph{ground types} we define
\begin{itemize}
\item $\sem{m}=\Meas\{1,\ldots,m\}$ where $\{1,\ldots,m\}$ is equipped with the discrete $\sigma$-algebra. Note that $\sem{m}\simeq \R^m$, and thus $\sem{\tUnit}\simeq\R$, the unit of the positive projective tensor.
\item $\sem{\tInt}=\Meas\N$, where $\N$ is equipped with the discrete $\sigma$-algebra
\item $\sem{\tReal}=\Meas\R$, where $\R$ is equipped with its usual Borel $\sigma$-algebra
\item $\sem{\tPosDef}=\Meas\PosDef(n)$, where $\PosDef(n)$ is the space of positive semi-definite  $n\times n$ matrices equipped with the Borel $\sigma$-algebra inherited from $\R^{n\times n}$
\end{itemize}

As expected, the tensor and function type constructors are interpreted by the monoidal closed structure of $\Roban$, i.e.
\begin{align*}
\sem{\Type[S]\otimes \Type}:=\sem{\Type[S]}\pptp\sem{\Type}\qquad\text{ and } \qquad\sem{\Type[S]\to\Type}:=\left[\sem{\Type[S]},\sem{\Type}\right]
\end{align*}
The higher-order probability type constructor $\tMeas$ is interpreted as follows. For any regular ordered Banach space $V$ we consider the underlying set together with the Borel $\sigma$-algebra induced by the norm. We then apply the functor $\Meas$ to this measurable space. This construction is functorial and we overload $\Meas$ to denote the resulting regular ordered Banach space by $\Meas V$. Using this convenient notation we define
\[
\sem{\tMeas\Type}:=\Meas\sem{\Type}
\]

For \emph{Bayesian types} note that the type system in Fig. 1 can only produce a Bayesian type $(\Type,\mu)$ if $\Type$ is a measure type and $\mu$ has no free variables, i.e. if $\context[\emptyset]\mu:\Type$ is derivable. We will therefore only need to provide a semantics to Bayesian types of this shape. Our semantics of Bayesian types is in some respect similar to that of \emph{pointed types} used in homotopy type theory \cite{licata2014eilenberg}. Indeed, at the type-theoretic level they are defined identically as a type together with a term inhabiting this type. However, the ordered vector space structure allows us to provide a semantics which is much richer than a space with a distinguished point. Given a measure type $\Type$ and a sequent of the type $\context[\emptyset]\mu:\Type$, we will see in \cref{sec:semanticsterms} that $\mu$ is interpreted as an operator $\sem{\mu}:\R\to\sem{\Type}$, which is uniquely determined by $\sem{\mu}(1)$. For notational clarity we will often simply write $\mu$ for the measure $\sem{\mu}(1)$. We define the denotation of the Bayesian type $(\Type,\mu)$ as the \emph{principal band in $\sem{\Type}$ (see \cref{sec:bands}) generated by the measure $\mu$} (i.e.\@$\sem{\mu}(1)$). Formally:
\begin{align}
\sem{(\Type,\mu)}=\sem{\Type}_{\mu}\label{eq:semBayesian}
\end{align}
For this semantics to be well-defined it is necessary that $\sem{\Type}$ be a Riesz space, since bands are defined using the lattice structure. This
is indeed the case:

\begin{theorem}\label{Th:meastypes}
The semantics of a measure type is a Banach lattice.
\end{theorem}

The function type constructor is the only operation in the type system which forces us to leave the category of Banach lattices and enter the much larger category $\Roban$. As shown in \cite[Ex. 1.17]{aliprantis2006positive}, the space of regular operators between two Riesz spaces need not even be a lattice. The non-closure of Banach lattices under taking internal homs is one of the technical reasons for our use of `measure types' \footnote{Note that we could in principle extend \eqref{eq:semBayesian} to all types by considering subsets generated by a single element which exist in all regular ordered Banach spaces, for example the closure of ideals.}.

We introduced order-complete types in \cref{sec:types} because of a `dual' non-closure property: order-complete spaces are not closed under the positive projective tensor operation.  As shown in \cite[4C]{fremlin1974tensor} the product $\Lps[2](\left[0,1\right])\pptp\Lps[2](\left[0,1\right])$ is not order-complete, even though $\Lps[2](\left[0,1\right])$ is. Order-completeness will be important in the semantics of \texttt{while} loops.

\begin{theorem}\label{Th:ordercomplete}
The semantics of an order-complete type is an order-complete space.
\end{theorem}

\subsubsection*{\uline{Subtypes and contexts}}
The subtyping relation will simply be interpreted as subspace inclusion. For example the relation $(\Type,\mu)<:\Type$ is interpreted as the inclusion of the principal band $\sem{\Type}_{\sem{\mu}}\hookrightarrow\sem{\Type}$.
A context $\tGamma$ will be interpreted as the positive projective tensor
\[
\sem{\tGamma}=\bigpptp_{\hspace{-2ex}i}^{\sup \supp(\tGamma)} \sem{\tGamma(i)}.
\]
and  we put $\sem{\emptyset}:=\R$. A typing rule $\context e:\Type$ will be interpreted as a regular (in fact positive, see Th. \ref{Th:programs}) operator
$\sem{e}:\sem{\tGamma}\to\sem{\Type}$ .

\subsection{Semantics of well-formed expressions}\label{sec:semanticsterms}

Let us now turn to the semantics of terms. 

\subsubsection{\uline{Constants}} A constant $c\in G$ whose ground type $\Type[G]$ is interpreted as the space $\Meas G$ will be interpreted as the operator 
\[
\sem{c}:\sem{\emptyset}=\R\longrightarrow\sem{\Type[G]}=\Meas G,\quad \lambda\mapsto\lambda\delta_c
\]
\subsubsection{\uline{Built-in operations}} Recall that every built-in operation $\tt{op}$ comes with typing information $(\Type[G]_1,\ldots,\Type[G]_n,\Type)$ where each $\Type[G]_i$ is of ground type and $\Type$ is either of ground type or of type $\tMeas\Type[G]$. Each such operation is interpreted via a function
$f_{\tt{op}}: G_1\times \ldots\times G_n\to X$, with $X=G$ or $X=\Meas G$, as the unique regular operator which linearizes $\Meas f_{\tt{op}}\circ \times$ according to the universal property \eqref{diag:tensorUP} of $\pptp$:
\[
\xymatrix@C=8ex@R=3ex
{
\Meas G_1\times \ldots\times\Meas G_n\ar[d]_{\times}\ar[r]_{\pptp} & \sem{\Type[G]_1}\pptp\ldots\pptp\sem{\Type[G]_n}\ar@/^1pc/@{-->}[ddl]^{\sem{\tt{op}}} \\
\Meas(G_1\times \ldots\times G_n) \ar[d]_{\Meas f_{\tt{op}}}\\
\Meas X
}
\]
For example the boolean operator \texttt{or} of type $(\tBool,\tBool,\tBool)$  would be interpreted, via the function $f_{\tt{or}}: 2\times 2\to 2$ implementing the boolean join, as the linearisation of $\Meas f_{\tt{or}}\circ \times$ (which is bilinear). Similarly, the operation $\tt{normal}$ of type $(\tReal,\tPosDef[1],\tMeas\tReal)$ building a normal distributions would be interpreted, via the obvious function $f_{\tt{normal}}:\R\times\R^+\to\Meas\R$, as the linearisation of $\Meas f_{\tt{normal}}\circ\times$. Note that if the inputs are deterministic, i.e. a tensor $\delta_\mu \otimes\delta_\sigma$ for a mean $\mu$ and a standard deviation $\sigma$ (as would usually be the case), then $\sem{\tt{normal}}(\delta_\mu \otimes\delta_\sigma)$ outputs \emph{a Dirac delta} over the distribution $\mathcal{N}(\mu,\sigma)$. Note how we interpret the deterministic construction of a distribution over $X$ differently from sampling an element of $X$ according to this distribution: the former is a distribution over distributions, the latter just a distribution. 

\subsubsection{\uline{Variables and assignments}} A variable on its own acts like a variable declaration and introduces a context (see Fig. 1). Its semantics is simply given by the identity operator on the type of the variable, formally if $[i\mapsto\Type]\vdash x_i:\Type$ then $\sem{x_i}=\mathrm{Id}_{\sem{\Type}}$. 
In order to define the semantics of variable assignment we need the following result.\footnote{As a consequence of this theorem, the semantics of all our types are \emph{Archimedean} ordered vector spaces.}

\begin{theorem}\label{Th:poslinfunc}
The denotation of any type $\Type$ admits a strictly positive functional $\phi_{\sem{\Type}}$.
\end{theorem}

The strictly positive functional $\phi_{\sem{\Type}}$ can be thought of as a generalisation to all types of the functional on measures which consists in evaluating the mass of the whole space, i.e. of $ev_X: \Meas X\to\R, \mu\mapsto\mu(X)$. With this notion in place we can provide a semantics to assignments. Given a sequent $\context e:\Type$, let $\phi_{\sem{\Type}}$ be the strictly positive functional on $\sem{\Type}$ constructed in Th. \ref{Th:poslinfunc} and let us write $\sem{\tGamma[i\mapsto\Type]}$ as $\sem{\tGamma_1}\pptp\sem{\Type}\pptp\sem{\tGamma_2}$. We now define the multilinear map 
\begin{align*}
&\sem{\tGamma_1}\times\sem{\Type}\times\sem{\tGamma_2}\longrightarrow\sem{\Type}\\
&(\gamma_1,t,\gamma_2)\mapsto\begin{cases}
\phi_{\sem{\Type}}(t) \sem{e}(\gamma_1\pptp t\pptp\gamma_2) &\text{if }\tGamma(i)=\Type\\
\phi_{\sem{\Type}}(t) \sem{e}(\gamma_1\pptp\gamma_2) &\text{else }
\end{cases}
\end{align*}
This defines the unique linearizing operator\footnote{In fact a \emph{nuclear operator} \cite{abramsky1999nuclear}.}
\[
\sem{x_i:=e}:\sem{\tGamma[i\mapsto\Type]}\longrightarrow \sem{\Type}
\]

In the case where the context $\tGamma$ is empty, the premise of the typing rule for assignments is interpreted as an operator $\sem{e}: \R\to\sem{\Type}$ and we can therefore strengthen the definition of $\sem{x_i:=e}$ as follows:
\[
\sem{x_i:=e}: \sem{\Type}\to\sem{\Type}_{\sem{e}(1)},\quad t\mapsto \phi_{\sem{\Type}}(t)\sem{e}(1)
\]
In the empty context, the general rule for variable assignment is a consequence of the rule creating Bayesian types since $(\Type,e)<:\Type$, i.e.\@ there is no disagreement between the two rules.

As a simple example, it is easy to type-check the program \texttt{x:=3.5} and see by unravelling the definition that it is interpreted as the operator $\Meas \R\to (\Meas\R)_{\delta_{3.5}}, \mu\mapsto \mu(\R)\delta_{3.5}$. In particular any probability distribution gets mapped to $\delta_{3.5}$. This is the semantics of assignment of \cite{K81c}.

\subsubsection{\uline{Sequencing and sequential composition}} These are conceptually straightforward as they essentially implement some form of function composition. The only difficulty resides in the bookkeeping of contexts which is a bit cumbersome. 

Given $\sem{e_1}:\sem{\tGamma}\to\sem{\Type[S]}$ and $\sem{e_2}: \sem{\tt{\Delta} [i\mapsto\Type[S]]}\to \sem{\Type}$  we define
$N=\sup(\supp(\tGamma)\cup\supp(\tt{\Delta}))$, we assume w.l.o.g. that $i\notin\supp(\tGamma)$ and define the semantics of $\tt{let}~x_i=e_1~\tt{in}~e_2$ as the unique operator which linearizes the multilinear map
\begin{align*}
&\prod_{i=1}^N \sem{(\tGamma\oplus\tt{\Delta})(i)} \to \sem{\Type}\\
& (x_1,\ldots,x_N) \mapsto e_2\left(\bigotimes_{i\in \supp(\tt{\Delta})}x_i,e_1\left(\bigotimes_{i\in\supp(\tGamma)}x_i\right)\right)
\end{align*}
The disjointness condition on the contexts $\tGamma$ and $\tt{\Delta}$ implements the resource-awareness of the system: $e_1$ consumes the context $\tGamma$, so no part of it can be re-used in the computation $e_2$.  

Sequential composition is just a generalisation of sequencing where $e_1$ outputs a store and the `cut' can take place over more than one type. As a simple example consider the program $\tt{x_1:=3.5~;~x_2:=7.3}$. Using the sequential composition rule we can derive the following typing-checking proof
\begin{center}
\footnotesize
\AxiomC{$\emptyset\vdash 3.5:\tReal$}
\UnaryInfC{$[1\mapsto \tReal]\vdash \tt{x_1:=3.5}:(\tReal,3.5)$}
\AxiomC{$\emptyset\vdash 7.3:\tReal$}
\UnaryInfC{$[2\mapsto \tReal]\vdash \tt{x_2:=7.3}:(\tReal,7.3)$}
\BinaryInfC{$[1\mapsto \tReal,2\mapsto \tReal]\vdash\tt{x_1:=3.5~;~x_2:=7.3}:(\tReal,3.5)\otimes (\tReal,7.3) $} 
\DisplayProof
\end{center}
The semantics of the program is the operator defined by
\begin{align*}
&\Meas\R\pptp\Meas\R\to (\Meas\R)_{\delta_{3.5}}\pptp(\Meas\R)_{\delta_{7.3}},\\ &\mu\otimes \nu\mapsto \mu(X)\delta_{3.5}\otimes\nu(X)\delta_{7.3}
\end{align*}

\subsubsection{\uline{$\lambda$-abstraction and function application}} These are interpreted exactly as expected in a monoidal closed category, namely via the adjunction $-\pptp\sem{\Type[S]} \dashv [\sem{\Type[S]},-]$ and ordinary function application.
\begin{remark}
While the denotation of $\lambda$-abstraction is immediately given by the monoidal closed structure of $\Roban$, the following point is worth making. Assume a context of ground types only. In the system of \cite{K81c}, such a context is of the shape $\Meas(X_1\times \ldots\times X_n)$, i.e. \emph{any} joint distribution over the variables can be considered as an input to the program. If we were Cartesian closed, a context would be of the shape $\Meas X_1\times \ldots\times\Meas X_n$, i.e. \emph{only product distributions} would be considered as potential inputs to the program. Our semantics lies somewhere in between these two possibilities since a context is of the shape $\Meas X_1\pptp\ldots\pptp\Meas X_n$. This means that not all joint probabilities can be $\lambda$-abstracted on, only those which live in the tensor product (i.e. limits of Cauchy sequences of linear combinations of product measures). Put differently, we can only $\lambda$-abstract if the probabilistic state of the machine is prepared (to use a quantum analogy) to a distribution in the positive projective tensor product.
\end{remark}
\subsubsection{\uline{Conditionals and $\tt{while}$ loops}} We provide the semantics of conditionals, the semantics of $\tt{while}$ loops then follows exactly as in \cite{K81c}: by first writing the fixpoint equation in terms of conditionals and then solving it using the side condition that $\sem{\Gamma}$ is order-complete.

Given a boolean test $\context e:\tBool$ interpreted as an operator $\sem{e}:\sem{\tGamma}\to\Meas 2$, the order-completeness of $\sem{\Gamma}$ allows us to define the maps
\begin{align*}
& T_e:\sem{\tGamma}^+\hspace{-1ex}\to\sem{\tGamma}^+, \gamma \mapsto \bigwedge\{0\leq \gamma'\leq \gamma:\sem{e}(\gamma')(1)=\sem{e}(\gamma)(1)\}\\
& F_e:\sem{\tGamma}^+\hspace{-1ex}\to\sem{\tGamma}^+, \gamma \mapsto \bigwedge\{0\leq \gamma'\leq \gamma:\sem{e}(\gamma')(0)=\sem{e}(\gamma)(0)\}
\end{align*}

\begin{proposition}\label{prop:TeFe}
The maps $T_e$ and $F_e$ are additive and linear over $\R^+$.
\end{proposition}

Since regular ordered Banach spaces have a generating cone, we can uniquely extend $T_e$ and $F_e$ to the entire space $\sem{\tGamma}$ and define the semantics of the conditional $\context \tt{if} \enskip e_1 \enskip \tt{then}\enskip e_2\enskip\tt{else}\enskip e_3:\Type$ as the operator
\[
\sem{\tGamma}\to\sem{\Type}, \gamma\mapsto \sem{e_2}\circ T_{e_1}(\gamma)+\sem{e_3}\circ F_{e_1}(\gamma)
\]

To see why this definition makes sense we will briefly show that it recovers the semantics of \cite{K81c}. In \cite{K81c}, $\sem{\tGamma}$ is a measure space $\Meas X$ and $\sem{e_1}:\Meas X\to\Meas 2$ is of the shape $\Meas b$ for a measurable map $b: X\to 2$ which specifies a measurable subset $B$ of $X$. We claim that $T_e: \Meas X\to\Meas X$ sends a probability measure $\mu$ to the measure $\mu_B$ defined by $\mu_B(A)=\mu(A\cap B)$, exactly as the operator $e_B$ of \cite[3.3.4]{K81c}. By unravelling the definition we want to show that
\[
\mu_B=\bigwedge\{0\leq \nu\leq \mu : \nu(B)=\mu(B)\}
\] 
Note first that $\mu_B$ belongs to the set above, so it remains to show that it is its minimal element. Let $\nu$ also belong to this set, and let $A$ be a measurable set. We decompose $A$ as $A=(A\cap B)\uplus(A\cap B^c)$. By definition
\[
\mu_B(A\cap B^c)=0\leq \nu(A\cap B^c)\text{ since }0\leq \nu.
\]
Moreover we have
\[
\mu_B(A\cap B)=\mu(A\cap B)=\nu(A\cap B).
\]
For if we had $\nu(A\cap B)<\mu(A\cap B)$, then in order to keep $\mu(B)=\nu(B)$ we would need $\nu(A^c\cap B)>\mu(A^c\cap B)$, a contradiction with $\nu\leq \mu$. Thus $\mu_B\leq \nu$ as claimed and the semantics of $\tt{if} \enskip e_1 \enskip \tt{then}\enskip e_2\enskip\tt{else}\enskip e_3:\Type$ becomes the operator
\[
\mu\mapsto \sem{e_2}(\mu_B) + \sem{e_3}(\mu_{B^c})
\]
exactly as in \cite{K81c}.

\subsubsection{\uline{$\tt{sampler}$ and $\tt{sample}$}} are given a semantics which can be understood as generalisations of the unit and co-unit of the Giry monad \cite{giry1982categorical} respectively. First we need the following easy result.

\begin{theorem}\label{Th:meastypesmeas}
The semantics of every measure type is isometrically and monotonically embedded in a space of measures $\Meas X$.
\end{theorem}

We can now define the semantics of $\tt{sampler}$. Suppose we have $\sem{e}:\sem{\tGamma}\to\sem{\Type}$ and, by Th. \ref{Th:meastypesmeas}, that $\sem{\Type}$ is isometrically and monotonically embedded in the space $\Meas X$. Now consider the map  $\eta: X\to\Meas X, x\mapsto \delta_x$ (which is \emph{not} an operator) and define denotation of $\tt{sampler}(e)$ as the positive operator
\[
\sem{\tt{sampler}(e)}:\sem{\tGamma}\to\Meas\sem{\Type}, \gamma \mapsto\Meas\eta(\gamma).
\] 

The semantics of \texttt{sample} works in the opposite direction. Suppose we have $\sem{e}:\sem{\tGamma}\to\Meas \sem{\Type}$ with $\sem{\Type}$ isometrically and monotonically embedded in $\Meas X$, then each element of $\Meas \sem{\Type}$ is also an element of $\Meas\Meas X$. We can define a map
\begin{align*}
m_X: ~ &\Meas\Meas X\to\Meas X,\\ 
&\rho\mapsto \lambda B. \int_{B^+(\Meas X)} ev_B(\mu) ~d\rho
\end{align*}
where we recall that $B^+(\Meas X)$ is the positive unit ball of the space $\Meas X$. This map clearly defines a positive operator, and in the case  where $\rho$ is supported by the set of probability distributions, i.e. the shell $\{\mu\in(\Meas X)^+:\norm{\mu}=1\}$ of the positive unit ball, $m_X$ coincides with the multiplication of the Giry monad. We are now ready to define
\[
\sem{\tt{sample}(e)}:\sem{\tGamma}\to\sem{\Type}, \gamma\mapsto m_X(\gamma).
\]
We can now interpret the type of the small Gaussian inference program of \cref{sec:gaussianex}. In defining the semantics of built-in operations we saw that the semantics of $\tt{\context[\emptyset]normal(0,1):\tMeas\tReal}$ is the linear map $\R\to\Meas\Meas\R$ mapping $1$ to the Dirac delta over the normal distribution $\mathcal{N}(0,1)$. It follows that
\[
\sem{\tt{sample(normal(0,1))}}(1)=\mathcal{N}(0,1)
\]
and by unravelling the definition we similarly find that
\begin{align*}
&\sem{\tt{sample(normal(sample(normal(0,1)),1)}}(1)=\\
&\lambda B. \int_{x\in\R}\mathcal{N}(x,1)(B)~d\mathcal{N}(0,1)=\mathcal{N}(0,\sqrt{2})
\end{align*}
which is the pushforward of $\mathcal{N}(0,1)$ by the kernel $\lambda x.~\mathcal{N}(x,1)$ described in \eqref{eq:kernelpushforward}. It follows that the output of the Gaussian inference program is interpreted as a linear operator
\[
(\Meas\R)_{\mathcal{N}(0,\sqrt{2})}\longrightarrow (\Meas\R)_{\mathcal{N}(0,1)}.
\]
We now turn to the semantics of the observe statement, which will show us what this operator actually is.

\subsubsection{\uline{Semantics of \texttt{observe}}}\label{sec:semobserve} 
Assume that we have $\sem{e}: \sem{(\Type[S],\mu)}\to \sem{\Type}$ with $\Type[S],\Type$ of measure type, i.e. we can type $\sem{e}$ as a regular operator $(\Meas X)_\mu \to \Meas Y$.

We now make the assumption, which we justify in Th. \ref{Th:programs} below, that $\sem{e}$ is positive. Under this assumption if $\nu\leq K\mu$, then $\sem{e}(\nu)\leq K\sem{e}(\mu)$, i.e. $\sem{e}$ restricts to an operator
\[
\sem{e}: (\Meas X)_{\mu}^{UB}\to(\Meas Y)_{\sem{e}(\mu)}^{UB}
\]
where $(\Meas X)_{\mu}^{UB}$ is the set of measures uniformly bounded by a multiple of $\mu$ (see \cite{ampba}).

The semantics of $\tt{observe}(e)$ is then fundamentally contained in the K\"{o}the dual operator
\[
\sem{e}\kd: \left((\Meas Y)_{\sem{e}(\mu)}^{UB}\right)\kd \to \left((\Meas X)_{\mu}^{UB}\right)\kd.
\]
It is not hard to check by using the Riesz Representation and Functional Representations natural transformations ($\RR$ and $\FR$ in Diagram \eqref{diag:summary}) that $\left((\Meas Y)_\nu^{UB}\right)\kd\simeq (\Lps(Y,\nu))\kd\simeq (\Meas Y)_\nu$, and thus $\sem{e}\kd$ can, modulo these isomorphisms, be typed as an operator
\begin{equation}\label{eq:observeSem}
\sem{e}\kd: (\Meas Y)_{\sem{e}(\mu)} \to (\Meas X)_\mu
\end{equation}
which is what the typing rule for \texttt{observe} requires. 

To illustrate how this semantics really implements the Bayesian inversion described in \cref{sec:BayesianInverse}, let us again consider our simple Gaussian inference program. The underlying Bayesian model is given by the probability kernel $\mathcal{N}(-,1):\R\to\mMeas \R$
and the prior $\mathcal{N}(0,1)$ on $\R$. Together these define a $\Krn$-arrow $(\R,\mathcal{N}(0,1))\to(\R,\mathcal{N}(0,\sqrt{2}))$ which is implemented by the program
\[
[\tt{x}\mapsto (\tReal,\tt{normal(0,1)})]\vdash\tt{sample(normal(x,1))}:\tReal
\]
whose denotation is the positive operator
\[
\Meas_{-}(\mathcal{N}(-,1)):(\Meas\R)_{\mathcal{N}(0,1)}\to \Meas \R.
\] 
Using the same argument as above, we can restrict this operator as follows 
\[
\Meas_{-}(\mathcal{N}(-,1)):(\Meas\R)^{UB}_{\mathcal{N}(0,1)}\to (\Meas\R)^{UB}_{\mathcal{N}(0,\sqrt{2})}.
\]
As stated above, all the information about the semantics of $\tt{observe(sample(normal(x,1)))}$ is contained in the K\"{o}the dual of this operator, which, through the Riesz Representation and Functional Representations natural transformation, can be typed modulo isomorphism as
\[
(\Meas_{-}(\mathcal{N}(-,1)))\kd: (\Meas\R)_{\mathcal{N}(0,\sqrt{2})}\to (\Meas\R)_{\mathcal{N}(0,1)}.
\]
Using the other half of diagram \eqref{diag:summary}, that is to say the Radon-Nikodym and Measure Representation natural transformations ($\RN$ and $\MR$ in diagram \ref{diag:summary}), this operator is equal, modulo isomorphism, to the operator
\[
\Meas_{-}(\mathcal{N}(-,1)\dg): (\Meas\R)_{\mathcal{N}(0,\sqrt{2})}\to (\Meas\R)_{\mathcal{N}(0,1)}.
\]
Here the Bayesian inverse of our original probability kernel appears explicitly, showing that our semantics indeed captures the notion of Bayesian inverse.

There is one final subtlety which we need to account for. Given $\sem{e}:\sem{(\Type[S],\mu)}\to\sem{\Type}$, the typing rule for \texttt{observe} in fact makes the whole semantics described above parametric in a choice of measure absolutely continuous w.r.t. the prior $\mu$ (see Fig. 1). This is a simple technicality: morally and practically the parameter will always be set to the prior itself, in which case we get as output of the program the K\"{o}the dual described by \eqref{eq:observeSem}. Mathematically however, we can choose as input any $\nu\ll\mu$; the output operator is then defined by the following tortuous journey (similar to constructions in \cite{ampba})
\[
\xymatrix
{
(\Meas Y)_{\sem{e}(\mu)}\ar[r]^{\simeq} & \left((\Meas Y)_{\sem{e}(\mu)}^{UB}\right)\kd\ar[r]^{i\kd} & \left((\Meas Y)_{\sem{e}(\nu)}^{UB}\right)\kd\ar[d]_{\sem{e}\kd} \\
(\Meas X)_{\mu} & (\Meas X)_{\nu}\ar[l]_{j} & \left((\Meas X)_{\nu}^{UB}\right)\kd \ar[l]_{\simeq}
}
\]
where $i,j$ are the obvious inclusions.

\begin{remark}
The semantics of \texttt{observe} via the K\"{o}the dual is more general than a semantics in terms of Bayesian inversion/disintegration. Nothing prevents the introduction of ground types which stand for measurable spaces in which disintegrations do not exist. However, the K\"{o}the dual will still exist. Thus our semantics is free of some of the `pointful' technicalities surrounding the existence of disintegrations, and follow the `pointless' perspective advocated in \cite{2017:FoSSaCS}. Similarly, we do not have to worry about the ambiguity cause by the fact that disintegrations are only defined up to a null set: the K\"{o}the dual of an operator between regular ordered Banach spaces exists completely unambiguously. 
\end{remark}

\subsection{Some properties of the semantics}
The development of the semantics in the previous section has built-in soundness:
\begin{theorem}
The semantics is sound w.r.t. to the typing rules of Fig. 1.
\end{theorem}
More importantly, we can extend \cite[Th. 3.3.8]{K81c} by a straightforward induction and show that:
\begin{theorem}\label{Th:programs}
The semantics of any program is a positive operator of norm $\leq 1$.
\end{theorem}
However another result of \cite{K81c}, namely that the denotation of a program is entirely determined by its action on point masses \cite[Th. 6.1]{K81c} does \emph{not} hold any more. The reason is interesting and is worth a few words. It is immediate from the type system (Fig. 1) and the denotation of Bayesian types that the domain of the semantics of an \texttt{observe} statement may \emph{not contain any point masses at all}. For example in the case of the Gaussian inference program of \ref{sec:gaussianex}, this domain is $(\Meas\R)_{\mathcal{N}(0,1)}$ which contains no point masses at all.

\subsection{Comparison with semantics \`{a} la Scott.}\label{sec:Scott}

There are interesting parallels to be drawn between our semantics and the Scott-Strachey semantics in terms of domains. 

Looking at ground types first, it is worth noting that just like the \emph{flat domain} functor turns a set (of integers for example) into a valid semantic object (a domain), so the functor $\Meas$ turns a measurable set into a valid semantic object (a regular ordered Banach space). Similarly, just like the flat domain functor can turn a partial map between sets into a total map between domains, so the functor $\Meas$ turns a partial measurable map into a linear operator. Partiality is encoded by the presence of the bottom element in the case of domain, and by the possibility to lose mass (i.e. get subdistributions) in the case of spaces of measures.

We do not know yet if the semantics of every program in our language is $\sigma$-order continuous, which would be the equivalent of Scott-continuous in our setting. The fundamental difference however is that our semantic category is not Cartesian closed, but monoidal closed.

\section*{Acknowledgment}
The authors would like to thank Ilias Garnier for bringing \cite{min1983exponential} to their attention.

%
%
\bibliographystyle{IEEEtran}
\bibliography{bib/dahlqvist-kozen}

\newpage

\section*{Appendix}
\subsection{Proofs}

\begin{proof}[Proof of Proposition \ref{prop:regbounded}]
By the triangle inequality it is enough to reason about positive operators. The proof is by contradiction. Suppose that $f: U\to V$ is positive but not bounded, then we can find a sequence $\{x_n\}$ in $U$ with $\norm{x_n}=1$ such that $\norm{f(x_n)}\geq n^3$. By the axioms R1 and R2 and the positivity of $f$ we can assume w.l.o.g. that $x_n\geq 0$. Since $\sum_{n=1}^\infty \frac{\norm{x_n}}{n^2}$ converges, we have by completeness of Banach spaces that $\sum_{n=1}^\infty \frac{x_n}{n^2}$ converges to an element $x$ of $U$. Since every $x_n\geq 0$, it follows that $0\leq \frac{x_n}{n^2}\leq x$. Since positive operators are automatically monotone it follows from axiom R1 that:
\[
n=\frac{n^3}{n^2}\leq \norm{f\left(\frac{x_n}{n^2}\right)}\leq \norm{f(x)}
\]
for every $n\in\N$, a contradiction. Thus $f$ is (norm-) bounded and therefore continuous.
\end{proof}

\begin{proof}[Sketch of the proof of Theorem \ref{Th:Kothe}]
The proof of \cite[Th. 20.2, Cor. 20.3]{zaanen2012introduction} also holds when the domain space is a regular ordered space because the only property of the domain being used is that the positive cone is generating. Thus the set $V^\sim$
 of regular functionals on $V$ forms an order-complete Riesz space. Similarly, the proof of \cite[Th. 4.74]{aliprantis2006positive} holds when the domain space is a regular ordered space. Thus $V^\sim$ forms a Banach space as well. From \cite[Th. 22.2]{zaanen2012introduction} we know that the set $V\kd$ is a band in $V^\sim$. Now assume that $v_n$ is a Cauchy sequence in $V\kd$. We know that it converges to a $v\in V^\sim$. By \cite[Th. 15.6]{zaanen2012introduction} every norm convergent  sequence in a Banach lattice has a subsequence converging in order. Let $u_n$ be this sequence, i.e. $u_n$ converges in order to $v$. Since $V^\sim$ is order complete we can further assume by the $\liminf$ construction that $u_n$ is increasing, i.e. $u_n\uparrow v$. But since $V\kd$ is band, it must follow that $v\in V\kd$ as desired.
\end{proof}

\begin{proof}[Proof of Theorem \ref{Th:pteL1L1}]
By Th. \ref{Th:RN} $(\Meas X)_\mu\simeq \Lps[1](X,\mu)$ and $(\Meas Y)_\nu\simeq \Lps[1](Y,\nu)$. For each $u\in \Lps[1](X,\mu)\ptp\Lps[1](Y,\nu)$, let $f_u\in\Lps[1](X,\Lps[1](Y,\nu),\mu)$ be the associated Bochner $\mu$-integrable function given by Th \ref{Th:pte:L1}. This map in turns defines $\phi_u: X\times Y\to \R$ by $\phi_u(x,y)=f_u(x)(y)$. It is easy to see that $\phi_u\in \Lps[1](X\times Y,\mu\times \nu)$:
\[
\int_{X\times Y} \phi(x,y)~d(\mu\times \nu)=\int_X \int_Y f_u(x)(y) ~d\nu~d\mu
\]
By Bochner's integrability theorem \cite[11.44]{aliprantis} $f_u$ is Bochner $\mu$-integrable iff $\norm{f_u}_1$ is Lebesgue integrable which shows that the integral above is finite. The fact that the operation $u\mapsto f_u\mapsto\phi_u$ is an isometry is obvious, and it has an inverse which associates to $\phi\in\Lps[1](X\times Y,\mu\times \nu)$ the Bochner $\mu$-integrable function $f_\phi: X\to\Lps[1](Y,\nu) , x\mapsto \phi(x,\cdot)$.
\end{proof}

\begin{proof}[Proof of Theorem \ref{Th:pte:MM}]
By Th. \ref{Th:pte:M} $\Meas{X}\ptp\Meas{Y}$ can be embedded in $\Meas{(X,\Meas{Y})}$.
This allows us to define an embedding $J: \Meas{X}\ptp\Meas{Y}\to \Meas{(X\times Y)}$ as follows: for each $u\in \Meas{X}\ptp\Meas{Y}$, let $\mu_u$ denote the corresponding vector-valued  measure in $ \Meas{(X,\Meas{Y})}$, we can then define a measure on $X\times Y$ by putting for any rectangle $A\times B$ of $X\times Y$
\[
J(\mu_u)(A\times B):= \mu_u(A)(B)
\]
It is easy to check that $\norm{J}=1$.
\end{proof}

\begin{proof}[Proof of Theorem \ref{Th:meastypes}]
By induction on the structure of the type. For ground type, it follows from the fact that their denotations are spaces of the shape $\Meas X$ which is always a Banach lattice (see \cref{sec:banlattices}). The same argument holds in the inductive case if the outermost constructor is $\tMeas$. For tensor product it follows from the fact that the positive projective tensor product of two Banach lattice is a Banach lattice \cite{fremlin1974tensor}. Finally, any band in a Banach lattice is a Banach lattice (see \cref{sec:bands}).
\end{proof}

\begin{proof}[Proof of Theorem \ref{Th:ordercomplete}]
By induction on the structure of the type. We start with the first layer of the grammar \eqref{eq:completetypes}. Ground types are interpreted as spaces of the shape $\Meas X$ which are AL spaces, and thus order-complete by Th. \ref{Th:ALcomplete}. Bayesian types $(\Type[G],\mu)$ built from ground types are of the shape $\Meas X_\mu$, and thus isomorphic to an $\Lps[1]$-space as was shown in Ex. \ref{ex:Band}. These are AL spaces, and are thus order-complete. The case of positive projective tensors of the shape $(\Type[G],\mu)\otimes (\Type[G],\mu)$ then follows from Th. \ref{Th:pteL1L1}.
Finally, for the outer layer of the syntax, it is clear that whatever $\Type[S]$ is, since $\sem{\tMeas \Type[S]}=\Meas\sem{\Type[S]}$, we get an order-complete space. The case of internal homs follows from the well-known result from the theory of Riesz space which states that if $V$ is an order-complete space, so is $[U,V]$ (for a net $\{T_\alpha\}$ in $[U,V]$, the join is defined at each $u\in U^+$ as $\bigvee_\alpha T_\alpha(u)$ by order-completeness of $V$, this is map is linear and extends to an operator $U\to V$ by the fact that regular ordered spaces have a generating cone, for more details see the proof of \cite[Th. 1.18]{aliprantis2006positive}).
\end{proof}

\begin{proof}[Proof of Theorem \ref{Th:poslinfunc}]
By induction on the structure of the types. For the base case note that all ground types are interpreted as spaces of the shape $\Meas G$, on which the evaluation function $ev_G: \Meas G\to\R, \mu\mapsto \mu(X)$ is defined and is a strictly positive functional. The same clearly holds for types of the shape $\tMeas\Type$. Suppose now that $\sem{\Type}$ has a strictly positive functional, then for any $\mu:\Type$, since $\sem{(\Type,\mu)}$ is a subspace of $\sem{\Type}$ is clearly inherits this strictly positive functional. Now assume both $\sem{\Type[S]}$ and $\sem{\Type}$ have strictly positive functionals $\phi$ and $\psi$ respectively. Then the map $\sem{\Type[S]}\times \sem{\Type}\to\R$ defined by $(t,s)\mapsto \phi(t)\phi(s)$ is bilinear and strictly positive. It follows that there exists a strictly positive functional $\phi\pptp\psi$ on $\sem{\Type[S]}\pptp\sem{\Type}=\sem{\Type[S]\otimes\Type}$. For internal homs we proceed as follows: since $[\sem{\Type[S]},\sem{\Type}]$ is generated by positive operators, let us first consider a strictly positive operator $T$ and define 
\[
\chi(T)=\sup\{\psi(T(s))\mid s\in \sem{\Type[S]}^+\setminus\{0\}, \phi(s)\leq 1\}
\]
(note how this definition mimics the definition of the operator norm). Since $\phi$ is strictly positive, the set over which the supremum is taken is not empty, and since $T$ and $\psi$ are strictly positive, $0<\chi(T)<\infty$ \cite[p. 545]{davies1968structure}. Moreover, since the supremum of a set of sums equals the sum of the suprema, $\chi(T+S)=\chi(T)+\chi(S)$, and similarly since scalar multiplication distributes over suprema $\chi(\lambda T)=\lambda\chi(T)$, it follows that $\chi$ is linear on strictly positive operators. We can then extend $\chi$ to all regular operators by putting $\chi(T):=\chi(T^+)-\chi(T^-)$.
\end{proof}

\begin{proof}[Proof of Proposition \ref{prop:TeFe}]
To see that it is additive, note first that since $\sem{\tGamma}$ is order-complete it is in particular a Riesz space, and it therefore has the Riesz decomposition property \cite[8.9]{aliprantis}, which means that if $0\leq z\leq x+y$ there exist $0\leq z_1,z_2$ such that $z_1+z_2=z$ and $z_1\leq x, z_2\leq y$. From this and the linearity of $\sem{e}$ it now follows easily that $T_e$ and $F_e$ are additive. The linearity poses no problem.
\end{proof}

\begin{proof}[Proof of Theorem \ref{Th:meastypesmeas}]
By induction on the structure of the type, with the base case being tautological. Similarly, the case of Bayesian types and types of the shape $\tMeas\Type$ are also tautological. The case of tensor product types follows from Th. \ref{Th:pte:MM} and Th. \ref{Th:ppte:ALspaces} (the monotonicity of the embedding is obvious).
\end{proof}

\begin{proof}[Proof of Theorem \ref{Th:programs}]
The proof is by induction on the derivation tree of the program (see Fig. 1). The result holds trivially for constants, variables and subtyping. For built-in operations it follows from the fact that the pushforward operation has norm 1. For assignments, one can easily show that the norm of the strictly positive functional built in Th. \ref{Th:poslinfunc} and used in defining the semantics has norm 1 and assignment therefore also has norm 1. The case of sequencing and sequential composition follows from the fact that the composition of operators of norm $\leq 1$ has norm $\leq 1$. For $\lambda$-abstraction, the result follows from the fact that the universal bilinear map has norm 1, and the result is trivial for function application. For conditionals it follows from the definition of $T_e$ and $F_e$ and the fact that the norm is monotone in a regular ordered Banach space. For \texttt{while} loops the result follows from \cite{K81c}. Finally, for the observe statement the result follows from the fact that the dual of an operator of norm $\leq 1$ has norm $\leq 1$.
\end{proof}

\subsection{Background material on measure theory and Lebesgue integration.}

\subsubsection{\uline{Measures}}\label{sec:meas}
A \emph{measurable space} is a set $X$ equipped with a collection $\salg$ of subsets---called the \emph{measurable subsets}---which can intuitively be understood as the observable parts of $X$, and are therefore also referred to as \emph{events} in the probabilistic literature. The collection $\salg$ of measurable sets must contain the empty set and be closed under complementation and countable union. It follows from the de Morgan laws that $\salg$ is also closed under countable intersection. Any collection of subsets satisfying these closure properties is called a \emph{$\sigma$-algebra}.\footnote{The ``$\sigma$'' in $\sigma$-algebra refers to ``countable unions'' in the same way as $F_\sigma$ sets are countable unions of closed sets in descriptive set theory, with $\sigma$ standing for the German \textit{Summe}, union.}
A $\sigma$-algebra is thus an $\omega$-complete Boolean algebra of sets. An important example of a measurable space is the set $\R$ together with its \emph{Borel $\sigma$-algebra}, the smallest $\sigma$-algebra containing the open subsets of $\R$ with the usual topology. In general, the \emph{Borel sets} of a topological space are the smallest $\sigma$-algebra containing the open sets.

A \emph{signed measure} on a $\sigma$-algebra $\salg$ over $X$ is a map $\mu:\salg\to[-\infty,+\infty]$ associating to every event $B$ a ``weight'' $\mu(B)$ satisfying (i) $\mu(\emptyset)=0$; (ii) $\mu$ can assume at most one of the values $-\infty$ and $+\infty$; and (iii) $\mu$ is \emph{$\sigma$-additive}, that is, $\mu(\cup_{i=1}^\infty A_i)=\sum_{i=1}^\infty \mu(A_i)$ for any countable pairwise disjoint collection $(A_i)_{i\in\N}$ of measurable sets. A signed measure is called a \emph{measure} if it assumes values in $[0,\infty]$, a \emph{finite signed measure} if it assumes values in $(-\infty,\infty)$ (equivalently, if $\absv{\mu(X)}<\infty)$, and a \emph{probability measure} if it is a measure and $\mu(X)=1$.

The study of signed measures can largely be reduced to the study of measures as the following result shows.
\begin{theorem}[Hahn-Jordan decomposition {\cite[\S5.6.1]{dudley2002real}}]\label{Th:HahnJordan}
Every signed measure $\mu$ has a unique decomposition as a difference $\mu=\mu^+-\mu^-$ of two measures, at least one of which is finite.
\end{theorem}
\noindent The \emph{total variation measure} of a signed measure is defined as
\begin{align*}
\absv{\mu}(A)=\sup\left\{\sum_{n=1}^\infty \absv{\mu(A_n)}: \{A_n\}\subseteq\salg\text{ a partition of }A\right\}
\end{align*}
and the \emph{total variation} of $\mu$ is then defined as $\absv{\mu}(X)$.
If $\mu$ is a measure, then $\norm{\mu}=\mu(X)$. From Th. \ref{Th:HahnJordan} it we have
\[
\norm{\mu}=\sup_{A\in\salg} \absv{\mu(A)} + \absv{\mu(X\setminus A)}
\]
The last measure-theoretical definition we need  is the following: if   $\mu,\nu$ are measures on a $\sigma$-algebra $\salg$, $\mu$ is said to be \emph{absolutely continuous} w.r.t.\ $\nu$, notation $\mu\ll \nu$, if $\mu(B)=0$ whenever $\absv{\nu}(B)=0$.

A map $f: (X,\salg)\to (Y,\mathcal{G})$ between measurable spaces is called $\emph{measurable}$ if for every measurable subset $U\in\mathcal{G}$, $f\inv(U)\in\salg$. Given a such a measurable map and a measure $\mu$ on $X$, one defines the \emph{pushforward measure} $f_*(\mu)$ on $(Y,\mathcal{G})$ by $f_*(\mu)(B)=\mu(f\inv(B))$.

\subsubsection{\uline{Lebesgue integration} \cite[Ch. 4]{dudley2002real}}\label{sec:lebint} Given a measurable space $(X,\salg)$, a \emph{simple function} is any real-valued function $f: X\to\R$ expressible as a sum of the form
\begin{equation}\label{eq:simplefct}
f=\sum_{i=1}^n \alpha_i \one_{B_i}
\end{equation}
where $\alpha_i\in\R$, $B_i\in\salg$, and $\one_{B_i}$ is the characteristic function of $B_i$. A simple function can be equivalently expressed in many ways,  depending on the choice of $n$ and the $\alpha_i$ and $B_i$.
Given a signed measure $\mu$ on $\salg$, the \emph{integral of a simple function \eqref{eq:simplefct} w.r.t. to $\mu$} is given by 
\begin{equation}\label{eq:intsimplefct}
\int f~d\mu=\sum_{i=1}^n \alpha_i \mu(B_i)
\end{equation}
It can be shown that \eqref{eq:intsimplefct} is independent of the specific representation \eqref{eq:simplefct}. Note that \eqref{eq:intsimplefct} is linear as a function of $\mu$. It thus follows from Theorem \ref{Th:HahnJordan} that integration w.r.t. to signed measures is uniquely determined by integration w.r.t. to measures.
Similarly, \eqref{eq:intsimplefct} is \emph{linear} and \emph{positive} w.r.t. simple functions (positive simple functions have positive integrals). Since any simple function can be expressed as the difference of two positive simple functions, it follows that integrating simple maps w.r.t. to signed measures is completely determined by integrating positive simple maps w.r.t. to measures.
In this spirit, given a positive measurable function $f:X\to \R^+$, we define
\[
\int f~d\mu=\sup\left\{\int g~d\mu \mid 0\leq g\leq f, \hspace{1ex} g\text{ simple}\right\}
\]
For a general function $f: X\to \R$, we define $f^+(x)=\max(f(x),0)$ and $f^-(x)=-\min(f(x),0)$. Clearly $f^+, f^-$ are positive and $f=f^+-f^-$.
We now define the \emph{integral of $f$ w.r.t. to a signed measure $\mu$} as
\begin{align*}
\cint f~d\mu&=\cint f^+~d\mu - \cint f^-~d\mu\\
&=\cint f^+~d\mu^+-\cint f^+~d\mu^--\cint f^-~d\mu^-+\cint f^-~d\mu^-
\end{align*}

We shall return to the order-theoretic property of integration in \cref{subsec:OrdBan}. We conclude this summary of Lebesgue integration with one of the most important theorems in measure theory.
\begin{theorem}[Radon-Nikodym {\cite[Th 5.5.4]{dudley2002real}}]\label{Th:RN}
Let $(X,\salg)$ be a measurable space, $\mu$ a finite measure\footnote{The result holds more generally for any dominating $\sigma$-finite measure. We will only need finite signed measures in this paper.} on $X$, and $\nu$ a finite signed measure on $\salg$ such that $\nu\ll\mu$. Then there exists a measurable function $f: X\to \R$, unique up to a $\mu$-nullset, such that
\[
\nu(B)= \int \one_B.f~d\mu = \int_B f~d\mu 
\]
\end{theorem}
The function $f$ is called the \emph{Radon-Nikodym} derivative of $\mu$ w.r.t. $\nu$ and is denoted $\frac{d\nu}{d\mu}$. It was shown in \cite{MFPS2018} that the Radon-Nikodym derivative defines a natural transformation.

\subsection{Supplementary material on projective tensor products}

\subsubsection{\uline{Definition of a Banach space}} We only consider \emph{real} vector spaces in this paper. We therefore simply say `vector space' with the understanding that the scalar field is $\R$. A \emph{Banach space} is a vector space $V$ equipped with a norm $\norm{\cdot}: V\to[0,\infty)$ such that $V$ is complete for the metric induced by the norm, i.e.\ such that every Cauchy sequence in $V$ has a limit in $V$. The vector spaces $\R^n$ equipped with the usual Euclidean norm are Banach space. 

\subsubsection{\uline{Bochner integration}} One can generalise the ideas behind the Lebesgue integral to give a definition of the integral of a function taking its value in an arbitrary Banach space $V$. As in the case of the Lebesgue integral, we start with simple functions. Given a measurable space $(X,\salg)$, we generalise \eqref{eq:simplefct} and say that a function $f: X\to V$ is \emph{simple} if is expressible a sum

\[
f=\sum_{i=1}^n v_i\one_{B_i}
\]
where $v_i\in V, B_i\in\salg$. We define the \emph{Bochner integral} of this simple function as
\[
\int f~d\mu=\sum_{i=1}^n v_i\mu(B_i)
\]
that is to say the weighted average of the vectors $v_i$ with the weights $\mu(B_i)$. We will always assume that a Banach space comes equipped with its Borel $\sigma$-algebra. To extend Bochner integration to arbitrary measurable functions $X\to V$ is impossible in general since $V$, unlike $\R$, may have a very high dimension. Intuitively, it is in general not possible to approximate a function $f: X\to V$ with simple ones if the dimensionality of $V$ is too high. We thus restrict the definition of the Bochner integral to a limited class of measurable functions taking values in a separable subspace. Formally, a function $f: X\to V$ is \emph{$\mu$-essentially separately valued} if there exists $E\in\salg$ and a separable subspace $Y\subseteq V$ such that $\mu(E^c)=0$ and $f(E)\subseteq Y$. We can now state:
\begin{theorem}[Pettis Measurability Theorem {\cite[Prop 2.15]{ryan2013introduction}}]
Let $(X,\salg)$ be a measurable space, let $\mu$ be a finite measure on $\salg$, and let $f: X\to V$ be a function taking values in a Banach space. Then $f$ is Borel measurable and $\mu$-essentially separately valued iff there exists a sequence $(f_n)$ of simple functions converging $\mu$-a.e. to $f$, in which case $f$ is said to be $\mu$-measurable.
\end{theorem}

A $\mu$-measurable function $f: X\to V$ will be called \emph{Bochner $\mu$-integrable} if the sequence of simple functions $(f_n)$ converging $\mu$-a.e. to $f$ also satisfies
\[
\lim_{n\to\infty} \int \norm{f-f_n}~d\mu =0
\] 
where the integral is the ordinary Lebesgue integral. The \emph{Bochner integral} of $f$ is the defined as the vector
\[
\int f~d\mu=\lim_{n\to\infty} \int f_n~d\mu
\]
Following the example of Lebesgue integration, we define \emph{Lebesgue-Bochner space} $\Lps[1](X,V,\mu)$ as the Banach space of (equivalence classes) of Bochner $\mu$-integrable functions $f: X\to V$ equipped with the norm
\[
\norm{f}_1=\int \norm{f}_V~d\mu
\]

\begin{theorem}[\cite{ryan2013introduction} Ex. 2.19]\label{Th:pte:L1}
$\Lps[1](X,\mu)\ptp V\simeq \Lps[1](X,V,\mu)$ the space of Bochner $\mu$-integrable maps $f: X\to V$.
\end{theorem}

\subsubsection{\uline{Vector-valued measures}}
One can also characterise projective tensor products of the shape $\Meas(X)\ptp V$. For this we need a couple of definitions which generalise the notion of measure to Banach spaces. Given a measurable space $(X,\salg)$ and a Banach space $V$, a $V$-valued \emph{vector measure} on $\salg$ is a $\sigma$-additive map $\mu:\salg\to V$, where $\sigma$-additivity means that if $A_i$ is a countable sequence of pairwise disjoint measurable sets, then the series $\sum_i \mu(A_i)$ converges to $\mu(\bigcup_i A_i)$. The  \emph{variation norm} of a $V$-valued measure is defined exactly like in the scalar-valued case:
 \[
\norm{\mu}_1=\sup \set{\sum_{A\in\BB}\norm{\mu(A)}}{\text{$\BB$ is a measurable partition of $X$}}
\]
where the norm $\norm{\mu(A)}$ is of course taken in $V$. 

Vector-valued measures and Bochner integration are related by a Radon-Nikodym type construction which can briefly be described as follows. Given a scalar measure $\lambda\in\Meas(X,\salg)$ and a Bochner $\lambda$-integrable function $f: X\to V$, we can use $f$ as a `density' and define a $V$-valued measure $\mu$ in the obvious way by
\begin{equation}\label{eq:vectorDensity}
\mu(A)=\int_A f~d\lambda.
\end{equation}
The notion of absolute continuity for vector-valued measures is defined as follows: a $V$-valued measure $\mu$ on $(X,\salg)$ is absolutely continuous w.r.t. to a scalar measure $\lambda\in\Meas(X,\salg)$ is $\lambda(E)=0$ implies $\mu(E)=0$.\footnote{Note that $\mu(E)=0$ is the zero of $V$ whilst $\lambda(E)=0\in \R$.} It is now natural to ask whether the Radon-Nikodym Theorem generalises to vector-valued measures and Bochner integrable densities. The answer is usually negative (see \cite[Ex. 5.13]{ryan2013introduction} for an example), and we therefore start from \eqref{eq:vectorDensity} to justify the following \emph{definition}: a $V$-valued measure $\mu$ on $(X,\salg)$ has the \emph{Radon-Nikodym property} if it has bounded variation and if for every finite scalar measure $\lambda\in\Meas(X,\salg)$ there exists a Bochner $\lambda$-integrable function $f: X\to V$ such that \eqref{eq:vectorDensity} holds. We can now state a characterisation of tensor products of the shape $\Meas X\ptp V$.

\begin{theorem}[\cite{ryan2013introduction} Th. 5.22]\label{Th:pte:M}
$\Meas{X}\ptp Y$ is isometrically isomorphic to the Banach space of $Y$-valued measures with the Radon-Nikodym property together with the variation norm.
\end{theorem}

\end{document}